\documentclass[twoside,11pt]{article}
\usepackage{jair, theapa, rawfonts}
\usepackage{fix-cm}
\usepackage[T1]{fontenc}
\usepackage{xspace}
\usepackage{amssymb}
\usepackage{amsmath}
\usepackage{amsthm}

\usepackage{nicefrac}
\usepackage{tikz}
\usepackage{pgflibraryshapes}
\usetikzlibrary{positioning,arrows,shapes,snakes,calc,fit}
\usepackage{colortbl}
\usepackage{booktabs}
\usepackage{array}
\usepackage{graphicx}
\usepackage{xcolor}
\usepackage{enumitem}
\usepackage{tabularx}
\usepackage[nameinlink]{cleveref}
\usepackage{thm-restate}
\usepackage{calc}

\newcommand{\todo}[1]{}

\theoremstyle{definition}

\newtheorem{remark}{Remark}

\theoremstyle{plain}

\newtheorem{proposition}{Proposition}
\newtheorem{corollary}{Corollary}

\newcommand{\sd}{\ensuremath{\mathit{SD}}\xspace}
\newcommand{\dd}{\ensuremath{\mathit{DD}}\xspace}
\newcommand{\supp}{\mathrm{supp}}

\title{On the Indecisiveness of Kelly-Strategyproof\\ Social Choice Functions
}
\author{\name Felix Brandt \email brandtf@in.tum.de \\
	\name Martin Bullinger \email bullinge@in.tum.de \\
	\name Patrick Lederer \email ledererp@in.tum.de \\
	\addr Institut f\"{u}r Informatik \\ Technische Universit\"{a}t M\"{u}nchen\\
	Boltzmannstr. 3, 85748 Garching, Germany 
	}

\ShortHeadings{On the Indecisiveness of Kelly-Strategyproof Social Choice Functions}
{Brandt, Bullinger, \& Lederer}
\firstpageno{1} 

\newcommand{\seti}[2]{i\hspace{3pt}\in\hspace{3pt} \medmuskip=0mu\relax
	\thickmuskip=1mu\relax [#1\dots#2]}
\newcommand{\set}[2]{\medmuskip=0mu\relax
	\thickmuskip=1mu\relax [#1\dots#2]}

\newcolumntype{L}[1]{>{\raggedright\let\newline\\\arraybackslash\hspace{0pt}}m{#1}}
\newcolumntype{C}[1]{>{\centering\let\newline\\\arraybackslash\hspace{0pt}}m{#1}}
\newcolumntype{R}[1]{>{\raggedleft\let\newline\\\arraybackslash\hspace{0pt}}m{#1}}

\newenvironment{profile}{\medmuskip=0mu\relax
	\thickmuskip=1mu\relax
	\tabular}{\endtabular\smallskip}

\newcommand{\profilewidth}{\textwidth}

\begin{document}

\maketitle

\begin{abstract}
	Social choice functions (SCFs) map the preferences of a group of agents over some set of alternatives to a non-empty subset of alternatives.
	The Gibbard-Satterthwaite theorem has shown that only extremely restrictive SCFs are strategyproof when there are more than two alternatives. For set-valued SCFs, or so-called social choice correspondences, the situation is less clear. There are miscellaneous---mostly negative---results using a variety of strategyproofness notions and additional requirements. The simple and intuitive notion of Kelly-strategyproofness has turned out to be particularly compelling because it is weak enough to still allow for positive results. For example, the Pareto rule is strategyproof even when preferences are weak, and a number of attractive SCFs (such as the top cycle, the uncovered set, and the essential set) are strategyproof for strict preferences. In this paper, we show that, for weak preferences, only indecisive SCFs can satisfy strategyproofness. In particular, \emph{(i)} every strategyproof rank-based SCF violates Pareto-optimality, \emph{(ii)} every strategyproof support-based SCF (which generalize Fishburn's C2 SCFs) that satisfies Pareto-optimality returns at least one most preferred alternative of every voter, and \emph{(iii)} every strategyproof non-imposing SCF returns the Condorcet loser in at least one profile.
	We also discuss the consequences of these results for randomized social choice.
\end{abstract}

\section{Introduction}

Whenever a group of agents aims at reaching a joint decision in a fair and principled way, they need to aggregate their individual preferences using a social choice function (SCF).
SCFs are traditionally studied by economists and mathematicians, but have also come under increasing scrutiny from computer scientists who are interested in their computational properties or want to utilize them in computational multiagent systems \shortcite<see, e.g.,>{BCE+14a,Endr17a}. 

An important phenomenon in social choice is that agents misrepresent their preferences in order to obtain a more preferred outcome. An SCF that is immune to strategic misrepresentation of preferences is called strategyproof.
\shortciteA{Gibb73a} and \shortciteA{Satt75a} have shown that only extremely restrictive single-valued SCFs are strategyproof: either the range of the SCF is restricted to only two outcomes or the SCF always returns the most preferred alternative of the same voter.
Perhaps the most controversial assumption of the Gibbard-Satterthwaite theorem is that the SCF must always return a single alternative \shortcite<see, e.g.,>{Gard76a,Kell77a,Barb77a,DuSc00a,Nehr00a,BDS01a,ChZh02a,Tayl05a}. This assumption is at variance with elementary fairness conditions such as anonymity and neutrality. For instance, consider an election with two alternatives and two voters such that each alternative is favored by a different voter. Clearly, both alternatives are equally acceptable, but single-valuedness forces us to pick a single alternative based on the preferences only.

We therefore study the manipulability of \emph{set-valued} SCFs (or so-called \emph{social choice correspondences}). When SCFs return sets of alternatives, there are various notions of strategyproofness, depending on the circumstances under which one set is considered to be preferred to another. When the underlying notion of strategyproofness is sufficiently strong, the negative consequences of the Gibbard-Satterthwaite theorem remain largely intact \shortcite<see, e.g.,>{DuSc00a,BDS01a,ChZh02a,Beno02a,Sato14a}.\footnote{We refer to \shortciteA{Barb10a} and \shortciteA{BSS19a} for a more detailed overview over this extensive stream of research.} In this paper, we are concerned with a rather weak---but natural and intuitive---notion of strategyproofness attributed to \shortciteA{Kell77a}. Several attractive SCFs have been shown to be strategyproof for this notion when preferences are strict \shortcite{Bran11c,BBH15a}. These include the top cycle, the uncovered set, the minimal covering set, and the essential set. However, when preferences are weak, these results break down and strategyproofness is not well understood in general.

\shortciteA{Feld79a} has shown that the Pareto rule is strategyproof according to Kelly's definition, even when preferences are weak. Moreover, the omninomination rule and the intersection of the Pareto rule and the omninomination rule are strategyproof as well \shortcite[Remark 1]{BSS19a}. These results are encouraging because they rule out impossibilities using Pareto-optimality and other weak properties.\footnote{For example, \shortciteA{BSS19a} have shown that Pareto-optimality is incompatible with anonymity and a notion of strategyproofness that is slightly stronger than Kelly's.}
In the context of strategic abstention (i.e., manipulation by deliberately abstaining from an election), even more positive results can be obtained. \shortciteA{BBGH18a} have shown that all of the above mentioned SCFs that are strategyproof for strict preferences are immune to strategic abstention even when preferences are weak.

A number of negative results were shown for severely restricted classes of SCFs. \shortciteA{Kell77a} and \shortciteA{Barb77a} 
have shown independently that there is no strategyproof SCF that satisfies quasi-transitive rationalizability. However, this result suffers from the fact that quasi-transitive rationalizability is almost prohibitive on its own \shortcite<see, e.g.,>{MCSo72a}.\footnote{This is acknowledged by \shortciteA{Kell77a} who writes that ``one plausible interpretation of such a theorem is that, rather than demonstrating the impossibility of reasonable strategy-proof social choice functions, it is part of a critique of the regularity [rationalizability] conditions.''} In subsequent work by \shortciteA{MaPa81a} and \shortciteA{Band83c}, quasi-transitive rationalizability has been replaced with weaker conditions such as minimal binariness or quasi-binariness, which are still very demanding and violated by most SCFs. \shortciteA{Barb77b} has shown that positively responsive SCFs fail to be strategyproof under mild assumptions. However, positively responsive SCFs are almost always single-valued and of all commonly considered SCFs only Borda's rule and Black's rule satisfy this criterion. More recently, \shortciteA[Theorem~8.1.2]{Tayl05a} has proven that every SCF that returns the set of all weak Condorcet winners whenever this set is non-empty fails to be strategyproof. This result was strengthened by \shortciteA{Bran11c}, who showed that every SCF that uniquely returns the (strict) Condorcet winner whenever one exists fails to be strategyproof. \shortciteA{BSS19a} have shown with the help of computers that every Pareto-optimal SCF whose outcome only depends on the pairwise majority margins can be manipulated.

Note that---in contrast to most other work---the strong impossibility theorems by \shortciteA{Bran11c} and \shortciteA{BSS19a} require weak preference relations, i.e., these authors assume that preference relations are transitive and complete, but not necessarily anti-symmetric. We follow this approach because ties arise quite naturally in many applications. In fact, we see little justification to assume that all agents entertain strict preferences. For example, a voter who strongly cares about the environment may deem all parties that deny climate change equally unacceptable. Moreover, preferential voting rules are often criticized for being impractical because they put an unduly heavy burden on voters by asking them to submit a complete and strict ranking of, say, 20 alternatives. This burden can be reduced by allowing voters to express indifferences between similar alternatives. The case of indifferences is even more striking when the set of alternatives consists of partitions of agents or assignments of objects to agents. In these settings, agents are likely to be indifferent between coalitions or assignments in which they are grouped with the same agents or in which they receive the same objects. 

For these reasons, we study strategyproofness based on the assumption that voters can express indifferences between alternatives. In particular, we investigate three broad classes of SCFs: \emph{rank-based} SCFs (which include all scoring rules), \emph{support-based} SCFs (which generalize Fishburn's C2 SCFs), and \emph{non-imposing} SCFs (which return every alternative as the unique winner for some preference profile). An overview of the three classes and typical examples of SCFs belonging to these classes are given in \Cref{fig:SCFs}. The classes are unrelated in a set-theoretic sense: for any subset of classes, there exist SCFs which lie precisely in these classes. For instance, Borda's rule is contained in all three classes. Taken together, they cover virtually all SCFs commonly considered in the literature.
 
For rank-based and support-based SCFs, we show that Pareto-optimality and strategyproofness imply that every voter is a nominator, i.e., the resulting choice sets contain at least one most preferred alternative of every voter. In the case of rank-based SCFs, this entails an impossibility (\Cref{thm1}) whereas for support-based SCFs it demonstrates a high degree of indecisiveness in the sense that the SCF tends to return large choice sets (\Cref{thm2}). 
For non-imposing SCFs, we show that strategyproofness implies that the Condorcet loser has to be returned in at least one preference profile (\Cref{thm4}).
The latter result remarkably holds without imposing fairness conditions such as anonymity or neutrality and can again be phrased in terms of indecisiveness: every strategyproof SCF that satisfies the Condorcet loser property will never return certain alternatives alone. Hence, our main theorems can be summarized by the observation that strategyproofness requires an SCF to return unreasonably large choice sets for some preference profiles. 

All these results rely on two auxiliary statements (\Cref{lem:NU,lem:CL}) which discuss the relationship between decisive, nominating, and vetoing groups of voters. Roughly, these concepts address how much influence a group of voters has when acting unanimously, e.g., a group of voters is decisive if it can ensure a subset of its best alternatives to be chosen and vetoing if it can ensure its unique least preferred alternative not to be chosen.
Due to their universality, these lemmas may be of independent interest. Our results can also be interpreted in the context of randomized social choice (where the outcome is a lottery over the alternatives instead of a set of alternatives). In more detail, all our axioms can be transferred to the randomized setting, and thus we also derive strong impossibilities for randomized social choice. 

Even though our main results are rather negative, they are important to improve our understanding of strategyproof SCFs. Much more positive results are obtained by making minuscule adjustments to the assumptions such as restricting the domain of preferences to strict preferences, weakening the underlying notion of strategyproofness, or replacing strategic manipulation with strategic abstention \shortcite<see, e.g.,>{Nehr00a,Bran11c,BBGH18a}. In all of these cases, a small number of support-based Condorcet extensions such as the top cycle, the uncovered set, the minimal covering set, and the essential set constitute appealing positive examples.

\begin{figure}[tb]
	\centering
	\begin{tikzpicture}[scale=0.8, every node/.style={scale=0.8}]
	  \tikzset{venn circle/.style={draw,circle,minimum width=#1}}

	  \node [venn circle =6.5cm] (A) at (-.2,0) {};
	  \node [anchor=north, xshift = -1cm, yshift = 2cm] at (A) {$2$-plurality};
	  \node [xshift = -2.8cm, yshift = 3cm] at (A) {\textbf{Rank-based}};

	  \node [venn circle = 6.5cm] (C) at (0:3.7cm) {};
	  \node [anchor=north, xshift = 1cm, yshift = 2cm] at (C) {$2$-Copeland};
	  \node [xshift = 3cm, yshift = 3cm] at (C) {\textbf{Support-based}};
	  \node[anchor=north, align = center, yshift = 1.7cm] at (barycentric cs:A=1/2,C=1/2 ) {$2$-Borda\\[1ex] \emph{constant rules}};
	  
	  \node [draw,ellipse,minimum width=9cm, minimum height = 6cm] (B) at ([yshift=-2.5cm]$(A)!1/2!(C)$) {};
	  \node [yshift = -1.5cm, align = center] at (B) {Young\\ Dodgson\\[1ex] \emph{scoring runoff rules}};
	  \node [yshift = -3.35cm] at (B) {\textbf{Non-imposing}};
	  
	  \node[align = center, xshift = -1.6cm, yshift = -0.2cm] at (barycentric cs:A=1/2,B=1/2 ) {Plurality\\Omninomination\\[1ex] \emph{most scoring rules}};
	  \node[align = center, xshift = 1.5cm, yshift = 0cm] at (barycentric cs:B=1/2,C=1/2 ) {Pareto\\Copeland\\[1ex] \emph{common}\\ \emph{Condorcet extensions}};
	  \node at (barycentric cs:A=1/3,B=1/3,C=1/3 ){Borda};
	\end{tikzpicture}
	\caption{The classes of rank-based, support-based, and non-imposing SCFs and typical examples. $2$-plurality, $2$-Copeland, and $2$-Borda return all alternatives whose respective score is at least as large as the second-highest score. All scoring rules except Borda's rule are rank-based, non-imposing, but not support-based. Common Condorcet extensions include the top cycle, the uncovered set, the minimal covering set, the essential set, the Simpson-Kramer rule, Nanson's rule, Schulze's rule, and Kemeny's rule. We refer to \shortciteA[Chapters 2-5]{BCE+14a} for definitions of these SCFs.
	\label{fig:SCFs}}
\end{figure}
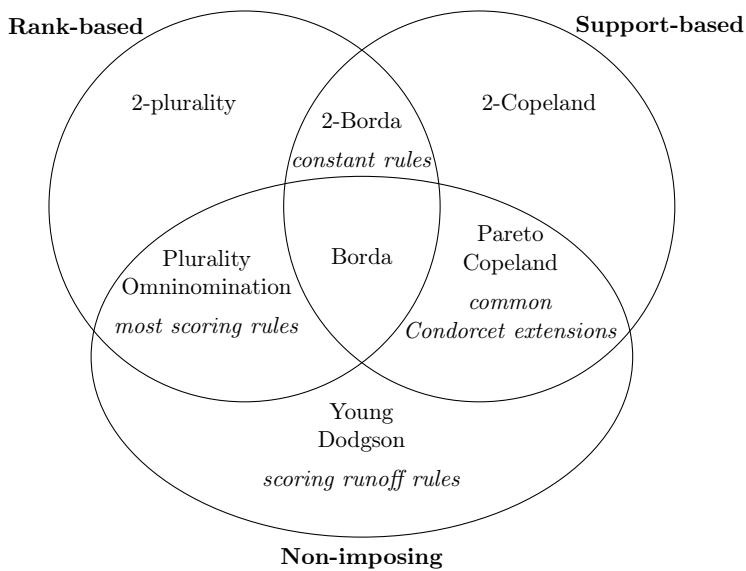

\section{The Model}
Let $N=\{1,\dots,n\}$ denote a finite set of voters and let $A=\{a,b,\dots\}$ denote a finite set of $m$ alternatives. Moreover, let ${\set{x}{y}} = \{i\in N\colon x\leq i\leq y\}$ denote the subset of voters from $x$ to $y$ and note that $\set{x}{y}$ is empty if $x>y$. Every voter $i\in N$ is equipped with a \emph{weak preference relation} $\succsim_i$, i.e., a complete and transitive binary relation on~$A$. We denote the strict part of $\succsim_i$ by $\succ_i$, i.e., $x\succ_i y$ if and only if $x\succsim_i y$ and $y\not\succsim_i x$, and the indifference part by $\sim_i$, i.e., $x\sim_i y$ if and only if $x\succsim_iy$ and $y\succsim_ix$. We compactly represent preference relations as comma-separated lists, where sets of alternatives express indifferences. For example, $x \succ y \sim z$ is represented by $x,\{y,z\}$. Furthermore, we call a preference relation $\succsim$ \emph{strict} if its irreflexive part is equal to its strict part $\succ$. The set of all weak preference relations on $A$ is called $\mathcal{R}$. A \emph{preference profile} $R\in\mathcal{R}^n$ is an $n$-tuple containing the preference relation of every voter $i\in N$. When defining preference profiles, we specify a set of voters who share the same preference relation by writing the set directly before the preference relation. For instance, $\set{x}{y}:$ $a,b,c$ means that all voters $\seti{x}{y}$ prefer $a$ to $b$ to $c$. We omit brackets for singleton sets.

Our central object of study are \emph{social choice functions (SCFs)}, or so-called social choice correspondences, which map preference profiles to non-empty sets of alternatives, i.e., functions of the form $f:\mathcal{R}^n\rightarrow 2^A\setminus\{\emptyset\}$.

\subsection{Axioms for Social Choice Functions}

We now introduce axioms that formalize desirable properties for SCFs, all of which are well-known in the literature. A basic fairness condition is anonymity, which requires that all voters are treated equally: an SCF $f$ is \emph{anonymous} if $f(R)=f(R')$ for all preference profiles $R$, $R'$ for which there is a permutation $\pi:N\rightarrow N$ such that ${\succsim'_i}={\succsim_{\pi(i)}}$ for all~$i\in N$. 

Perhaps one of the most prominent axioms in economic theory is Pareto-optimality, which is based on the notion of Pareto-dominance: an alternative $x$ \emph{Pareto-dominates} another alternative $y$ if $x\succsim_i y$ for all $i\in N$ and there is a voter $j\in N$ with $x\succ_j y$.
An alternative is \emph{Pareto-optimal} if it is not Pareto-dominated by any other alternative. This idea leads to the \emph{Pareto rule} which returns all Pareto-optimal alternatives. Finally, an SCF~$f$ is \emph{Pareto-optimal} if it never returns Pareto-dominated alternatives. 

An axiom that is closely related to Pareto-optimality is \emph{near unanimity}, as introduced by \shortciteA{Beno02a}. Near unanimity requires that $f(R)=\{x\}$ for all alternatives $x\in A$ and preference profiles $R$ in which at least $n-1$ voters uniquely top-rank $x$. The more voters there are, the more compelling near unanimity is.

A natural weakening of these axioms is \emph{non-imposition} which requires that for every alternative $x\in A$, there is a profile $R$ such that $f(R)=\{x\}$. For single-valued SCFs, non-imposition is almost imperative because it merely requires that the SCF is surjective. For set-valued SCFs, as considered in this paper, this is not necessarily the case. For example, every SCF that always returns at least two alternatives fails non-imposition (see, for example, $2$-plurality, $2$-Borda, and $2$-Copeland in \Cref{fig:SCFs}). 

An influential concept in social choice theory is that of a Condorcet winner, which is an alternative that wins all pairwise majority comparisons. For formally defining this term, let the pairwise \emph{support} of $x$ over $y$ denote the number of voters who strictly prefer $x$ to~$y$, i.e., $s_{xy}(R)=|\{i\in N\colon x\succ_iy\}|$. Then, an alternative $a\in A$ is a \emph{Condorcet winner} if $s_{ax}(R)>s_{xa}(R)$ for all $x\in A\setminus \{a\}$. An SCF is \emph{Condorcet-consistent} or a so-called \emph{Condorcet extension} if it uniquely returns the Condorcet winner whenever one exists. 

Analogously, an alternative~$a$ is called \emph{Condorcet loser} if $s_{xa}(R)>s_{ax}(R)$ for all $x\in A\setminus \{a\}$. An SCF $f$ satisfies the \emph{Condorcet loser property} if $x\not\in f(R)$ whenever $x$ is the Condorcet loser in $R$. While there are Condorcet extensions that violate the Condorcet loser property (e.g., the Simpson-Kramer rule) and SCFs that satisfy the Condorcet loser property but fail Condorcet-consistency (e.g., Borda's rule), the Condorcet loser property ``feels'' weaker. This follows the intuition that both properties affect exactly the same number of preference profiles, but the Condorcet loser property only excludes a single alternative (and otherwise leaves a lot of freedom) whereas Condorcet-consistency completely determines the (singleton) choice set.

\subsection{Strategyproofness}
One of the central problems in social choice theory is manipulation, i.e., voters may lie about their true preferences to obtain a more preferred outcome. For single-valued SCFs, it is clear what constitutes a more preferred outcome. In the case of set-valued SCFs, there are various ways to define a manipulation depending on the assumptions about the voters' preferences over sets of alternatives. Here, we make a simple and natural assumption first considered by \shortciteA{Kell77a}:
a voter $i$ weakly prefers a set $X$ to another set $Y$, denoted by $X\succsim_i Y$, if and only if $x\succsim_i y$ for all $x\in X, y\in Y$. Thus, the strict part of this preference extension is
\begin{align*}
	X~\succ_i~Y\text{ if and only if } & \text{for all }x\in X, y\in Y, x\succsim_iy \text{ and}\\&\text{there are }x'\in X, y'\in Y\text{ with }x'\succ_i y'\text.
\end{align*}

An SCF is manipulable if a voter can improve his outcome by lying about his preferences. Formally, an SCF $f$ is \emph{manipulable} if there are a voter $i\in N$ and preference profiles $R$, $R'$ such that ${\succsim_j}={\succsim_j'}$ for all $j\in N\setminus\{i\}$ and $f(R')~\succ_i~f(R)$. Moreover, $f$ is \emph{strategyproof} if it is not manipulable. 

These assumptions can, for example, be justified by considering a randomized tie-breaking procedure (a so-called lottery) that is used to select a single alternative from every set of alternatives returned by the SCF. 
We then have that $X\succ_i Y$ if and only if all lotteries with support~$X$ yield strictly more expected utility than all lotteries with support~$Y$ for all utility functions that are ordinally consistent with $\succsim_i$ \shortcite<see, e.g.,>{Gard79a,BSS19a}.

Note that, in the proofs of this paper, we often do not need the full power of strategyproofness. Instead, we mainly consider two types of manipulations: 
either the original choice set only consists of the manipulator's least preferred alternatives, or the new choice set only consists of the manipulator's most preferred alternatives.
In order to formalize these situations, we define $T_i(R)$ as the set of voter $i$'s top-ranked alternatives in $R$ and $B_i(R)$ as the set of voter $i$'s bottom-ranked alternatives in $R$. We then derive the following two consequences of strategyproofness, where $R$ and $R'$ are two preference profiles that only differ in the preference relation of voter $i$:

	\begin{enumerate}[label=\emph{(SP\arabic*)}, ref=\emph{SP\arabic*}, itemindent=0.25cm]
	    \item \label{imp1} If $f(R)\subseteq B_i(R)$, then $f(R')\subseteq B_i(R)$.
	    \item \label{imp2} If $f(R)\subseteq T_i(R')$, then $f(R')\subseteq T_i(R')$.
	\end{enumerate}
	
	\ref{imp1} states that, if a subset of voter $i$'s least preferred alternatives is the choice set for $R$, then the choice set after a manipulation is also a subset of voter $i$'s least preferred alternatives; any other outcome constitutes a manipulation for voter $i$. On the other hand, \ref{imp2} states that turning the current choice set $f(R)$ into a subset of voter $i$'s most preferred alternatives in $R'$ results into a choice set $f(R')$ that is a subset of $T_i(R')$. If this was not true, voter $i$ could manipulate $f$ by deviating from $R'$ to $R$.
	
	For a better understanding of strategyproofness, \ref{imp1}, and \ref{imp2}, consider the following profiles $R$ and $R'$ and assume that $f$ is a strategyproof SCF.\smallskip
	
	\begin{profile}{C{0.05\profilewidth} C{0.25\profilewidth} C{0.25\profilewidth} C{0.3\profilewidth}}
		{$R$:} & $1$: $\{a, b\}, c, \{d,e\}$ & $2$: $\{a, c\}, \{d,e\}, b$  & $3$: $d,c,a,b,e$\\
		{$R'$:} & $1$: $\{b,c\}, a, \{d,e\}$ & $2$: $\{a, c\}, \{d,e\}, b$  & $3$: $d,c,a,b,e$
	\end{profile}
	
	If $f(R)=\{d\}$, then \ref{imp1} requires that $f(R')\subseteq \{d,e\}$ since voter $1$ prefers every set $X$ with $X\not\subseteq \{d,e\}$ to $f(R)$. Furthermore, if $f(R)=\{c\}$, then \ref{imp2} implies that $f(R')\subseteq \{b,c\}$; otherwise voter $1$ can manipulate by reverting from $R'$ to $R$. Finally, note that strategyproofness is stronger than the conjunction of \ref{imp1} and \ref{imp2}, e.g., if $f(R)=\{c,d\}$, then $f(R')\neq \{b,c\}$ as voter $1$ could manipulate otherwise.

\subsection{Decisive, Nominating, and Vetoing Groups of Voters}

A common concern in social choice theory is that single voters or groups of voters might be more influential than others \shortcite<see, e.g.,>{LeWe11a}. Perhaps the most prominent example of such a notion are dictators: a voter $i\in N$ is a \emph{dictator} for an SCF $f$ if $f(R)$ always chooses a subset of voter $i$'s most preferred alternatives, i.e., if $f(R)\subseteq T_i(R)$ for all profiles $R$. The existence of dictators is usually undesirable because it means that a single voter can determine the outcome of the election alone. 

A related but far less restrictive concept concerns the notion of nominators: a voter $i\in N$ is a \emph{nominator} for an SCF $f$ if $f(R)$ always contains at least one of his most preferred alternatives. More formally, a voter $i\in N$ is a nominator for an SCF $f$ if $f(R)\cap T_i(R)\neq \emptyset$ for all preference profiles $R$. Nominators are weak dictators in the sense that they can always force an alternative into the choice set by reporting it as their unique top choice. 

Finally, we formalize the converse idea that a voter might be able to prohibit an alternative from being chosen. This leads to the notion of a \emph{vetoer}, which is a voter $i\in N$ such that $f(R)$ does never contain the uniquely least preferred alternative of voter $i$. If a vetoer does not report a uniquely least preferred alternative, the corresponding SCF is not restricted. Formally, a voter $i\in N$ is a vetoer for an SCF $f$ if $f(R)\cap B_i(R)=\emptyset$ for all preference profiles $R$ with $|B_i(R)|=1$. 

In the context of social choice, the existence of dictators, nominators, and vetoers is often undesirable as these notions formalize that some voters have an undesirably large impact on the outcome. For instance, if a voter is a nominator, he can force an alternative $x$ into the choice set even if all other voters agree that $x$ is the worst option. 
To avoid these problems, it is natural to consider generalizations of dictators, nominators, and vetoers to groups of voters. We opt for a rather weak generalization of these axioms and require that all voters in the group need to report the same preference relation to influence the SCF.

We say that a non-empty set of voters $I\subseteq N$ is \emph{decisive} if $f(R)$ is a subset of the best alternatives of the voters $i\in I$ for all profiles $R$ such that ${\succsim_i}={\succsim_j}$ for all $i,j\in I$. The concept of decisive groups of voters is best known from Arrow's proof of his impossibility theorem \shortcite{Arro51a}. Similarly, a non-empty set of voters $I\subseteq N$ is \emph{nominating} if $f(R)$ contains at least one of the most preferred alternatives of the voters $i\in I$ for all profiles $R$ such that ${\succsim_i}={\succsim_j}$ for all $i,j\in I$, and \emph{vetoing} if $f(R)$ excludes the uniquely least preferred alternative of the voters $i\in I$ for all such profiles. 

The notions of decisive, nominating, and vetoing groups of voters are far less restrictive than the corresponding single voter notions. More precisely, if a single voter fulfills any of these properties, then also every group containing this agent does. In fact, many desirable axioms even imply that sufficiently large groups of voters need to be decisive or vetoing. For instance, Pareto-optimality implies that the set of all voters is decisive, and near unanimity for strategyproof SCFs is equivalent to the requirement that every group of $n-1$ voters is decisive.\footnote{The claim on near unanimity may fail for manipulable SCFs as near unanimity only affects profiles where $n-1$ voters \emph{uniquely} top-rank the same alternative, whereas a decisive group $I$ affects the outcome for all profiles $R$ with ${\succsim_i}={\succsim_j}$ for all $i,j\in I$.} Furthermore, the Condorcet loser property implies that every group $I$ with $|I|>\frac{n}{2}$ is vetoing. As we will show, there are strong relationships between the notions of decisive, nominating, and vetoing groups for strategyproof SCFs. 

\subsection{Rank-Basedness and Support-Basedness}

In this section, we introduce two classes of anonymous SCFs that capture many of the SCFs commonly studied in the literature: rank-based and support-based SCFs. The basic idea of rank-basedness is that voters assign ranks to the alternatives and that an SCF should only depend on the ranks of the alternatives, but not on which voter assigns which rank to an alternative. In order to formalize this idea, we first need to define the rank of an alternative. In the case of strict preferences, this is straightforward: the rank of alternative~$x$ according to $\succsim_i$ is $\bar{r}(\succsim_i,x)=|\{y\in A\colon y\succsim_i x\}|$ \shortcite{Lasl96a}. By contrast, there are multiple possibilities how to define the rank in the presence of ties. We define a new and very weak notion of rank-basedness for weak preferences, making our results only stronger. To this end, define the \emph{rank tuple} of $x$ with respect to $\succsim_i$ as 
\begin{align*}
	r(\succsim_i,x)&=(\bar{r}(\succ_i,x), \bar{r}(\sim_i,x))\\
	&=(|\{y\in A\colon y\succ_i x\}|, |\{y\in A\colon y\sim_i x\}|)\text.
\end{align*}

The rank tuple contains more information than many other generalizations of the rank and therefore, it leads to a more general definition of rank-basedness. Next, we define the \emph{rank vector} of an alternative $x$ which contains the rank tuple of $x$ with respect to every voter in increasing lexicographic order, i.e., $r^*(R,x)=(r(\succsim_{i_1},x), r(\succsim_{i_2},x), \dots, r(\succsim_{i_n},x))$ where $\bar{r}(\succ_{i_j},x)\leq \bar{r}(\succ_{i_{j+1}},x)$ and if $\bar{r}(\succ_{i_j},x)= \bar{r}(\succ_{i_{j+1}},x)$, then $\bar{r}(\sim_{i_j},x)\leq \bar{r}(\sim_{i_{j+1}},x)$. Finally, the \emph{rank matrix} $r^*(R)$ of the preference profile $R$ contains the rank vectors as rows. An SCF $f$ is called \emph{rank-based} if $f(R)=f(R')$ for all preference profiles $R, R'\in \mathcal{R}^n$ with $r^*(R)=r^*(R')$. 
The class of rank-based SCFs contains many popular SCFs such as all scoring rules or the omninomination rule, which returns all top-ranked alternatives.\footnote{We refer to \shortciteA[Chapters 2-5]{BCE+14a} for the definitions of these and all following SCFs.}

A similar line of thought leads to support-basedness, which is based on the pairwise support of an alternative $x$ against another one $y$. Recall that the pairwise support refers to the number of voters who strictly prefer $x$ to $y$, i.e., $s_{xy}(R)=|\{i\in N\colon x\succ_iy\}|$. 
We define the \emph{support matrix} $s^*(R)=(s_{xy}(R))_{x,y\in A}$ which contains the supports for all pairs of alternatives. Finally, an SCF $f$ is \emph{support-based} if $f(R)=f(R')$ for all preference profiles $R, R'\in\mathcal{R}^n$ with $s^*(R)=s^*(R')$. Note that support-basedness is a new generalization of Fishburn's C2 to weak preferences \shortcite{Fish77a}. Hence, many well-known SCFs such as Borda's rule, Kemeny's rule, the Simpson-Kramer rule, Nanson's rule, Schulze's rule, the Pareto rule, and the top cycle are support-based. 

Support-basedness is less restrictive than \emph{pairwiseness}, which requires that $f(R)=f(R')$ for all preference profiles $R,R'\in\mathcal{R}^n$ with $s_{ab}(R)-s_{ba}(R)=s_{ab}(R')-s_{ba}(R')$ for all $a,b\in A$ \shortcite<see, e.g.,>{BSS19a}. For example, the Pareto rule is support-based, but fails to be pairwise. Another important subclass of support-based SCFs are majoritarian ones, which are merely based on the majority relation. To this end, we define the \emph{majority relation~$\succsim_R$} of a profile $R$ as ${\succsim_R}=\{(a,b)\in A^2\colon s_{ab}(R)\geq s_{ba}(R)\}$. Then, an SCF $f$ is \emph{majoritarian} if $f(R)=f(R')$ for all preference profiles $R$ and $R'$ with ${\succsim_R}={\succsim_{R'}}$ \shortcite<see, e.g.,>{BBGH18a}. For instance, the top cycle is majoritarian, whereas all other previous examples rely on the exact supports for computing the outcomes and thus fail this axiom.

In order to illustrate the definitions of rank-based, support-based, and majoritarian SCFs, we discuss two classical examples. First, consider the plurality rule, which returns all alternatives $x$ that maximize $|\{i\in N\colon \bar{r}(\succ_i, x)=0\}|$. By definition, this SCF is rank-based, but it is not support-based. The latter claim follows by considering the following preference profiles $R$ and $R'$ because $s^*(R)=s^*(R')$ but the plurality rule chooses $\{a\}$ for $R$ and $\{a,b\}$ for $R'$.\smallskip

	\begin{profile}{C{0.05\profilewidth} C{0.3\profilewidth} C{0.3\profilewidth} C{0.3\profilewidth}}
		{$R$:} & $1$: $\{a, b\}, c$ & $2$: $a,b,c$  & $3$: $c,b,a$\\
		{$R'$:} & $1$: $\{a, b\}, c$ & $2$: $a,c,b$  & $3$: $b,c,a$
	\end{profile}
	
As second example, consider the Pareto rule, which chooses all Pareto-optimal alternatives. This SCF is support-based because an alternative $x$ Pareto-dominates another alternative $y$ if and only if $s_{xy}(R)>0$ and $s_{yx}(R)=0$. On the other hand, it violates rank-basedness because there are profiles with the same rank matrix but different sets of Pareto-optimal alternatives (see Claim 1 in the proof of \Cref{thm1} for details). Finally, the Pareto rule is not majoritarian since it chooses $\{a\}$ for $\bar R$ and $\{a,b\}$ for $\bar R'$, but ${\succsim_{\bar R}}={\succsim_{\bar R'}}$.\smallskip

	\begin{profile}{C{0.05\profilewidth} C{0.3\profilewidth} C{0.3\profilewidth} C{0.3\profilewidth}}
		{$\bar R$:} & $1$: $\{a, b\},c$ & $2$: $\{a,b\},c$  & $3$: $a,b,c$\\
		{$\bar R'$:} & $1$: $b,a,c$ & $2$: $a,b,c$  & $3$: $a,b,c$
	\end{profile}

\section{Results}

The unifying theme of our results is that strategyproofness requires a high degree of indecisiveness. In more detail, we show that every voter is a nominator for all rank-based and support-based SCFs that satisfy Pareto-optimality and strategyproofness. Consequently, such SCFs have to choose a large number of alternatives for most preference profiles as one of the best alternatives of every voter needs to be in the choice set. 
For the very broad class of non-imposing SCFs, we show that every strategyproof SCF violates the Condorcet loser property. Put differently, for every strategyproof SCF that satisfies the Condorcet loser property, there is an alternative $x$ that is not returned as unique winner even if it is unanimously top-ranked.

In order to prove the claim for rank-based and support-based SCFs, we focus on its contrapositive, i.e., we assume that there is a rank-based or support-based SCF $f$ that satisfies Pareto-optimality and strategyproofness and for which a voter $i\in N$ is not a nominator. 
We first show that a group of voters is not nominating for a Pareto-optimal and strategyproof SCF if and only if its complement is decisive.

\begin{restatable}{lemma}{lemNU}\label{lem:NU}
Let $f$ be a Pareto-optimal and strategyproof SCF that is defined for $m\geq 3$ alternatives and $n\geq 2$ voters. A group of voters $I$ with $\emptyset \subsetneq I\subsetneq N$ is not nominating for $f$ if and only if $N\setminus I$ is decisive for $f$. 
\end{restatable}

\begin{proof}
	Let $f$ denote a Pareto-optimal and strategyproof SCF and consider an arbitrary set of voters $I$ with $\emptyset\subsetneq I\subsetneq N$. First, we show that $I$ is not nominating for $f$ if $N\setminus I$ is decisive. This follows immediately by considering a preference profile $R$ in which all voters in $N\setminus I$ report an alternative $a$ as best choice and are indifferent between the alternatives in $A\setminus \{a\}$, and the voters in $I$ report another alternative $b$ as their unique top choice and are indifferent between all alternatives in $A\setminus \{b\}$. Then, $f(R)=\{a\}$ because $N\setminus I$ is decisive for $f$, which proves that $I$ is not nominating as this condition requires that $b\in f(R)$.
	
	For the other direction, suppose that the group $I$ is not nominating for $f$. Our goal is to show that the group $N\setminus I$ is decisive for $f$ and we observe for this that there is a profile $R^0$ such that $f(R^0)\cap T_i(R^0)=\emptyset$ and ${\succsim_i^0}={\succsim_j^0}$ for all voters $i,j\in I$ because $I$ is not nominating for $f$. Subsequently, we will apply multiple transformations to $R^0$: first, we deduce a profile $R'$ such that $f(R')\cap T_i(R') = \emptyset$ for all $i\in I$ and $f(R')=\{x\}$ for some alternative $x\in f(R^0)$. Secondly, we infer from this profile that $f(R)=\{x\}$ for all preferences profiles $R$ in which the voters $j\in N\setminus I$ prefer $x$ uniquely the most. As third step, we generalize this observation from a single alternative to all alternatives. This is reminiscent of the so-called \emph{field expansion lemma} in proofs of Arrow's theorem \shortcite<see, e.g.,>{Sen86a}. Finally, we extend our analysis to the case where the voters in $N\setminus I$ may top-rank multiple alternatives. For an easier notation of the subsequent arguments, we assume that $I=\{1,\dots,k\}$ for some $k\in \{1,\dots, n-1\}$; this is without loss of generality because all our arguments are independent of the naming of the voters.\bigskip

	\textbf{Step 1:} As first step, we let the voters $j\in N\setminus I=\{k+1,\dots,n\}$ replace their preference relations in $R^0$ sequentially such that they prefer the alternatives in $f(R^0)$ the most. More formally, this means that we consider a sequence of preference profiles $R^{0,0}, \dots, R^{0,n-k}$ such that $R^{0,0}=R^0$ and $R^{0,i}$ evolves out of $R^{0,i-1}$ by assigning voter $k+i$ a preference relation such that $T_{k+i}(R^{0,1})=f(R^0)$. For each $i\in \{1,\dots, n-k\}$, \ref{imp2} implies that $f(R^{0,i})\subseteq f(R^0)$ if $f(R^{0,i-1})\subseteq f(R^0)$ because $f(R^{0,i-1})\subseteq T_{k+i}(R^{0,i})$. Since we start this process at the profile $R^0$, we derive a preference profile $R^1=R^{0,n-k}$ with $f(R^1)\subseteq f(R^0)$. 
	Next, let $a\in f(R^0)$ denote an alternative such that $a\succsim_i b$ for all $b\in f(R^0)$ and $i\in I$, i.e., $a$ is the most preferred alternative of the voters $i\in I$ in $f(R^0)$. Such an alternative exists since ${\succsim_i^1}={\succsim_j^1}$ for all $i,j\in I$. Moreover, let $B$ denote the set of alternatives such that $b\sim_i^1 a$ for all $b\in B$ and $i\in I$. As next step, we sequentially replace the preference relations of the voters $i\in I=\{1,\dots, k\}$ in $R^1$ with a preference relation where all alternatives in $T_i(R^1)$ are preferred to $a$, which, in turn, is preferred to all alternatives in $A \setminus (T_i(R^1) \cup \{a\})$. More formally, we consider again a sequence of preference profiles $R^{1,0},\dots, R^{1,k}$ such that $R^{1,0}=R^1$ and $R^{1,i}$ is derived from $R^{1,i-1}$ by modifying the preference relation of voter $i$ as described in the last sentence. Next, we will show for all $i\in \{1,\dots, k\}$ that if $f(R^{1,i-1})\subseteq B$, then $f(R^{1,i})\subseteq B$. Observe for this that, for all profiles $R^{1,i}$, alternative $a$ Pareto-dominates all alternatives $x\in A$ with $a\succ_j^1 x$ for all $j\in I$. Hence, Pareto-optimality ensures that $x\not\in f(R^{1,i})$ for all these alternatives. This means that if $f(R^{1,i-1})\subseteq B$, then $f(R^{1,i})\subseteq B$ because voter $i$ strictly prefers every Pareto-optimal alternative $x\in A\setminus B$ to all alternatives in $B$, i.e., every set of Pareto-optimal alternatives $X\not\subseteq B$ constitutes a manipulation for voter $i$. Finally, observe that $f(R^1)\subseteq B$ because $f(R^1)\subseteq f(R^0)$ and all alternatives $x\in f(R^0)$ with $a \succ_i^1 x$ are Pareto-dominated in $R^1$. We can therefore repeatedly apply the previous argument to derive that $f(R^2)\subseteq B$ for the profile $R^2=R^{1,k}$. Moreover, observe that in $R^2$, all voters $i\in I$ prefer $a$ to all alternatives $y\in A\setminus (T_i(R^0)\cup \{a\})$ and the voters $j\in N\setminus I$ top-rank $a$. Hence, $a$ Pareto-dominates all other alternatives in $B$, which implies that $f(R^2)=\{a\}$.\bigskip

	\textbf{Step 2:} Given the preference profile $R^2$ from the last step, we show that $f(R)=\{a\}$ for all preferences profiles $R$ in which the voters in $N\setminus I$ prefer $a$ uniquely the most. We deduce this result by modifying and analyzing the profile $R^2$. First, we sequentially change the preference relation of all voters $j\in N\setminus I$ such that they prefer $a$ uniquely the most and an alternative $b\in T_i(R^2)$ (for $i\in I$) uniquely the second most. Formally, this yields another sequence of preference profiles $R^{2,0},\dots, {R}^{2,n-k}$ such that ${R}^{2,0}=R^2$ and ${R}^{2,i}$ is derived from ${R}^{2,i-1}$ by making $a$ into the uniquely best alternative and $b$ into the uniquely second best alternative of voter $k+i$. Just as for the sequence $R^{0,i}$, \ref{imp2} implies that if $f({R}^{2,i-1})=\{a\}$, then $f({R}^{2,i})=\{a\}$ because $f({R}^{2,i-1})=\{a\}=T_{k+i}({R}^{2,i})$. Since $f(R^2)=\{a\}$, we infer for the profile $R^3={R}^{2,n-k}$ that $f(R^3)=\{a\}$ by repeatedly applying this argument. Furthermore, every alternative in $A\setminus \{a,b\}$ is Pareto-dominated by $b$ in $R^3$. We use this observation to sequentially replace the current preference relations of the voters $i\in I$ with a preference relation in which $b$ is the uniquely most preferred alternative and $a$ is his uniquely least preferred alternative. Formally, this leads to another sequence of profiles $R^{3,0},\dots, R^{3,k}$ that starts at $R^3$ and one by one changes the preference relation of the voters $i\in I$ as described. For every profile $R^{3,i}$, it holds that only $a$ and $b$ can be chosen because of Pareto-optimality. Moreover, if $f(R^{3,i-1})=\{a\}$, then $f(R^{3,i})=\{a\}$ as any other subset of $\{a,b\}$ constitutes a manipulation for voter $i$. Hence, this process results in a profile $R^4=R^{3,k}$ such that $f(R^4)=\{a\}$, all voters $i\in I$ prefer $a$ uniquely the least, and all voters $j\in N\setminus I$ prefer $a$ uniquely the most. It follows now from \ref{imp1} and \ref{imp2} that $f(R)=\{a\}$ for all preference profiles $R$ in which the voters $j\in N\setminus I$ prefer $a$ uniquely the most: \ref{imp1} allows the voters $i\in I$ to deviate to any other preference relation without changing the choice set because $B_i(R^4)=f(R^4)=\{a\}$ and \ref{imp2} allows the voters $i\in N\setminus I$ to reorder the alternatives in $A\setminus \{a\}$ arbitrarily because $f(R^4)=\{a\}$ is after the deviation still their set of top-ranked alternatives.\bigskip
	
	\textbf{Step 3:} As next step, we show that the voters in $N\setminus I$ can make every alternative win uniquely if they report it as their common top choice. Thus, consider the preference profile $R^5$ in which all voters in $N\setminus I$ prefer $a$ uniquely the most, and the voters in $I$ prefer~$c$ uniquely the most, $b$ uniquely second most, and $a$ uniquely the least. It follows from the last step that $f(R^5)=\{a\}$. Next, let the voters $j\in N\setminus I$ change their preferences sequentially such that they prefer $a$ and $b$ the most. Formally, this leads to another sequence of preference profiles $R^{5,0}, \dots, R^{5,n-k}$ and \ref{imp2} implies that if $f(R^{5,i-1})\subseteq \{a,b\}$, then $f(R^{5,i})\subseteq\{a,b\}$ because $T_i(R^{5,i})=\{a,b\}$. Hence, it holds for the profile $R^6=R^{5,n-k}$ that $f(R^6)=\{b\}$: our previous argument implies that $f(R^6)\subseteq \{a,b\}$ and $b$ Pareto-dominates~$a$ in this profile. Thereafter, we replace the preference relation of every voter $i\in N\setminus I$ with a new preference in which he prefers $b$ uniquely the most. \ref{imp2} shows for the corresponding sequence of profiles $R^{6,0}, \dots, R^{6,n-k}$ that $f(R^{6,i})=\{b\}$ if $f(R^{6,i-1})=\{b\}$. Therefore, this sequence results in a new preference profile $R^7=R^{6,n-k}$ with $f(R^7)=\{b\}$. Since the voters $i\in I$ do not top-rank $b$, we can now apply the constructions in Step 2 to deduce that $b$ is uniquely chosen if all voters in $N\setminus I$ voters prefer it uniquely the most.\bigskip
	
	\textbf{Step 4:} Finally, it remains to prove that $f(R)\subseteq T_i(R)$ for all voters $i\in N\setminus I$ and preference profiles $R$ such that ${\succsim_i}={\succsim_j}$ for all $i,j\in N\setminus I$. If the voters in $N\setminus I$ only report a single alternative $x$ as their top choice, this claim follows from Step 3, which requires that $x$ is the unique winner. Hence, consider a profile $R^8$ such that ${\succsim_i^8}={\succsim_j^8}$ for $i,j\in N\setminus I$, $|T_i(R)|\geq 2$ for all $i\in N\setminus I$, and let $a$ denote one of the top-ranked alternatives of these voters. Moreover, define $R^9$ as the profile derived from $R^8$ by making $a$ into the unique best alternative of every voter $i\in N\setminus I$. Step 3 implies for $R^9$ that $f(R^9)=\{a\}$. Moreover, we can go from $R^9$ to $R^8$ by letting the voters $i\in N\setminus I$ one after another revert back to the preference relation $\succsim^8_i$. Since all these voters have the same preference relation in $R^8$ and $a\in T_i(R^8)$ for all $i\in N\setminus I$, it follows from a repeated application \ref{imp2} that $f(R^8)\subseteq T_i(R^8)$, which proves the lemma. 
\end{proof}

\Cref{lem:NU} has a number of interesting consequences. First of all, it shows that, for every non-empty set of voters $I\subseteq N$, either $I$ is nominating or $N\setminus I$ is decisive for a strategyproof and Pareto-optimal SCF. Since anonymity implies that no set $I\subseteq N$ with $|I|\leq \frac{n}{2}$ can be decisive, it follows that every set with more than half of the voters is nominating for a Pareto-optimal, strategyproof, and anonymous SCF. Furthermore, this lemma shows that, under Pareto-optimality, strategyproofness, and anonymity, indecisiveness for a single preference profile of a particularly simple type entails a large degree of indecisiveness for the entire domain of preference profiles: if an alternative is not chosen uniquely even if $n-l$ voters prefer it uniquely the most, then all groups of size $l$ are nominating. This already indicates that strategyproof SCFs are rather indecisive under mild additional assumptions. For our subsequent proofs, the inverse direction of \Cref{lem:NU} is more interesting: if a voter $i$ is not a nominator, then the set $N\setminus \{i\}$ is decisive. This means that the absence of nominators implies near unanimity for strategyproof and Pareto-optimal SCFs. Furthermore, if we additionally assume anonymity, we have near unanimity even if a single voter is not a nominator.

\begin{remark}
Remarkably, many impossibility results rule out that every voter is a nominator. For instance, \shortciteA{DuSc00a}, \shortciteA{Beno02a}, and \shortciteA{Sato08a} invoke axioms prohibiting that every voter is a nominator. Moreover, a crucial step in the computer-generated proofs by \shortciteA[Theorem~3.1]{BBEG16a} and \shortciteA[Theorem~1]{BSS19a} is to show that there is a voter who is not a nominator. \Cref{lem:NU} shows that these assumptions and observations imply the existence of a decisive group of size $n-1$, which is in conflict with strategyproofness as defined by the above authors. Intuitively, a decisive group of size $n-1$ is already too small to allow for their notions of strategyproofness.
\end{remark}

\subsection{Rank-Based SCFs}

In this section, we prove that there is no rank-based SCF that satisfies Pareto-optimality and strategyproofness. This result follows from the observation that Pareto-optimality, strategyproofness, and rank-basedness require that every voter is a nominator, but Pareto-optimality and rank-basedness do not allow for such SCFs. 

It is possible to show \Cref{thm1}---as well as \Cref{thm2}---by induction proofs where completely indifferent voters and universally bottom-ranked alternatives are used to generalize the statement to arbitrarily many voters and alternatives \shortcite<see, e.g.,>{BBEG16a,BBGH18a,BSS19a}. Instead, we prefer to give universal proofs for any number of voters and alternatives to stress the robustness of the respective constructions. As a consequence, our proofs usually hold when restricting the domain of admissible profiles by prohibiting artificial constructs such as completely indifferent voters.
Note that we often assume that all voters are indifferent between all but a few alternatives $A\setminus X$. This assumption is not required and is only used for the sake of simplicity. In fact, the preferences between alternatives in $X$ can be arbitrary and may differ from voter to voter and often even between profiles. The only restriction is that the preferences involving alternatives in $A\setminus X$ are not modified.  

\begin{restatable}{theorem}{thmRB}\label{thm1}
There is no rank-based SCF that satisfies Pareto-optimality and strategyproofness if $m\geq 4$ and $n\geq 3$, or if $m\geq 5$ and $n\geq 2$.
\end{restatable}
\begin{proof}
	Consider fixed numbers of voters $n$ and alternatives $m$ such that $m\geq 4$ and $n\geq 3$, or $m\geq 5$ and $n\ge 2$. Furthermore, suppose for contradiction that there is a rank-based SCF $f$ that satisfies strategyproofness and Pareto-optimality for the given values of~$n$ and~$m$. We derive a contradiction to this assumption by proving two claims: on the one hand, there is a voter who is not a nominator for $f$. On the other hand, the assumptions on the SCF require that every voter is a nominator. These two claims contradict each other and therefore $f$ cannot exist.
	\medskip 
	
	\textbf{Claim 1: Not every voter is a nominator for $f$.}
	
	First, we prove that not every voter is a nominator for $f$. For this, we use a case distinction and first suppose that $m\geq 4$ and $n\geq 3$. In this case, consider the following three profiles, where $X=A\setminus\{a,b,c,d\}$. \smallskip

	\begin{profile}{C{0.05\profilewidth} C{0.25\profilewidth} C{0.25\profilewidth} C{0.3\profilewidth}}
		{$R^1$:} & $1$: $\{a, b\}, X, \{c, d\}$ & $2$: $\{c, d\}, X, \{a, b\}$  & $[3\dots n]$: $a, \{b, c, d\}, X$
	\end{profile}
	
	\begin{profile}{C{0.05\profilewidth} C{0.25\profilewidth} C{0.25\profilewidth} C{0.3\profilewidth}}
		{$R^2$:} & $1$: $\{a, c\}, X, \{b, d\}$ & $2$: $\{b, d\}, X, \{a, c\}$  & $[3\dots n]$: $a, \{b, c, d\}, X$
	\end{profile}
	
	\begin{profile}{C{0.05\profilewidth} C{0.25\profilewidth} C{0.25\profilewidth} C{0.3\profilewidth}}
		{$R^3$:} & $1$: $\{a, d\}, X, \{b, c\}$ & $2$: $\{b, c\}, X, \{a, d\}$  & $[3\dots n]$: $a, \{b, c, d\}, X$
	\end{profile}
	
	It can be easily verified that $r^*(R^1)=r^*(R^2)=r^*(R^3)$ and that $a$ Pareto-dominates $b$ in $R^1$, $c$ in $R^2$, and $d$ in $R^3$. This means that $f(R^1)=f(R^2)=f(R^3)\subseteq \{a\}\cup X$ because of rank-basedness and Pareto-optimality. Consequently, voter 2 is not a nominator for $f$.
	
	Next, we focus on the case that $m\geq 5$ and $n= 2$ and consider the profiles $R^4$, $R^5$, and $R^6$, where $X=A\setminus \{a,b,c,d,e\}$. \smallskip

		\begin{profile}{C{0.1\profilewidth} C{0.35\profilewidth} C{0.35\profilewidth}}
			$R^4$: & $1$: $\{a,b\},X,e,\{c,d\}$  & $2$: $\{c,d\},X,a,\{b,e\}$
		\end{profile}

		\begin{profile}{C{0.1\profilewidth} C{0.35\profilewidth} C{0.35\profilewidth}}
			$R^5$: & $1$: $\{a,c\},X,e,\{b,d\}$  & $2$: $\{b,d\},X,a,\{c,e\}$
		\end{profile}

		\begin{profile}{C{0.1\profilewidth} C{0.35\profilewidth} C{0.35\profilewidth}}
			$R^6$: & $1$: $\{a,d\},X,e,\{b,c\}$  & $2$: $\{b,c\},X,a,\{d,e\}$
		\end{profile}
		
		Analogous to the last case, it can be verified that $r^*(R^4)=r^*(R^5)=r^*(R^6)$, and that $a$ Pareto-dominates $b$ in $R^4$, $c$ in $R^5$, and $d$ in $R^6$. Consequently, rank-basedness and Pareto-optimality imply that $f(R^4)=f(R^5)=f(R^6)\subseteq \{a,e\}\cup X$, which proves that voter 2 is not a nominator for $f$. 
	\medskip
	
	\textbf{Claim 2: Every voter is a nominator for $f$.}
	
	Assume for contradiction that a voter is not a nominator for $f$ and let $X=A\setminus\{a,b,c,d\}$. Consequently, it follows from rank-basedness that no voter is a  nominator and therefore, \Cref{lem:NU} shows that all sets of $n-1$ voters are decisive. Furthermore, we want to point out that the subsequent construction works for all $n\geq 2$ and $m\geq 4$, which means that no case distinction is required. Our proof focuses on the profiles $R^{1,k}$ and $R^{2,k}$ for $k\in \{1,\dots, n\}$ shown below. \smallskip

	\begin{profile}{C{0.1\profilewidth} C{0.2\profilewidth} C{0.30\profilewidth} C{0.30\profilewidth}}
		{$R^{1,k}$:} & $1$: $\{c, d\}, X,b, a$  & $[2\dots k]$: $\{a, b\}, X,c, d$ & $[k+1\dots n]$: $a, X, b, c, d$
	\end{profile}
	
	\begin{profile}{C{0.1\profilewidth} C{0.2\profilewidth} C{0.30\profilewidth} C{0.30\profilewidth}}
		{$R^{2,k}$:} & $1$: $\{b,d\}, X,c, a$  & $[2\dots k]$: $\{a, b\}, X,c, d$  & $[k+1\dots n]$: $a, X, b, c, d$
	\end{profile}
	
	We prove by induction on $k\in\{1,\dots, n\}$ that $f(R^{1,k}) = f(R^{2,k}) = \{a\}$. The case $k = n$ yields a contradiction to Pareto-optimality as $a$ is Pareto-dominated by $b$ in $R^{1,n}$.
	
	The base case $k= 1$ follows because $n-1$ voters prefer $a$ uniquely the most in both $R^{1,1}$ and $R^{2,1}$. Therefore, our previous observation that every set of $n-1$ voters is decisive shows that $f(R^{1,1})=f(R^{2,1})=\{a\}$. 
	
	Assume now that the induction hypothesis is true for some fixed $k\in \{1,\dots,n-1\}$, i.e., $f(R^{1,k})=f(R^{2,k})=\{a\}$. 
	By induction and \ref{imp2}, $f(R^{1,k+1})\subseteq \{a,b\}$ since otherwise voter $k+1$ can manipulate by switching back to $R^{1,k}$. Next, we derive the profile $R^{3,k}$ shown below from $R^{2,k}$ by assigning voter $k+1$ the preference relation $\{a, c\}, X, b, d$.\smallskip
	
	\begin{profile}{C{0.1\profilewidth} C{0.35\profilewidth} C{0.35\profilewidth}}
		{$R^{3,k}$:} & $1:\,\{b, d\}, X, c, a$  & $[2\dots k]:\,\{a, b\}, X,c, d$  \\ & $k+1:\,\{a, c\}, X,b,d$ & $[k+2\dots n]:\,a, X, b, c, d$
	\end{profile}
	
	The induction hypothesis entails that $f(R^{2,k})=\{a\}$ and therefore, \ref{imp2} implies that $f(R^{3,k})\subseteq \{a,c\}$ because $f(R^{2,k})\subseteq T_{k+1}(R^{3,k})$. Next, we apply rank-basedness to conclude that $f(R^{1,k+1})=\{a\}$ as $r^*(R^{1,k+1})=r^*(R^{3,k})$. Finally, $R^{2,k+1}$ evolves from $R^{1,k+1}$ by having voter $1$ change his preferences. Since $B_1(R^{1,k+1})=\{a\}=f(R^{1,k+1})$, \ref{imp1} implies that $f(R^{2,k+1}) = \{a\}$ as any other outcome benefits voter $1$. This proves the induction step and therefore also the theorem.
\end{proof} 

\begin{remark}\label{rem:rankb}
The axioms used in \Cref{thm1} are independent: the Pareto rule satisfies all axioms except rank-basedness, the trivial SCF which always returns all alternatives only violates Pareto-optimality, and Borda's rule only violates strategyproofness.\footnote{We define Borda's rule as the SCF that chooses all alternatives $x$ that minimize $\sum_{i\in N} \bar{r}(\succ_i,x)+\frac{1}{2}\bar{r}(\sim_i,x)$. This definition of Borda's rule for weak preferences is equivalent to the one suggested by \shortciteA{Youn74a}.} Furthermore, the Pareto rule is rank-based if $m\leq 3$, and if $m= 4$ and $n\leq2$ (cf. \Cref{prop:POrb} in the appendix), which entails that the bounds on $m$ and $n$ are tight.
\end{remark}

\begin{remark}\label{rem:rankbvariant}
\Cref{thm1} is only an impossibility because of the lack of compatibility of rank-basedness and Pareto-optimality in Claim 1, independently of strategyproofness. By contrast, the main consequence of strategyproofness is indecisiveness as captured in Claim~2. Indeed, \Cref{thm1} breaks down once we weaken Pareto-optimality to weak Pareto-optimality (which only excludes alternatives for which another alternative is strictly preferred by every voter) as then the omninomination rule satisfies all required axioms \shortcite[Remark 6]{BSS19a}. By contrast, Claim 2 is rather robust since a number of variations are true: for instance, it is easy to adapt the proof of this claim to show that no neutral, strategyproof, and rank-based SCF satisfies near unanimity if $m\geq 4$ and $n\geq 3$, or $m\geq 5$ and $n=2$.\footnote{An SCF is neutral if $f(\pi(R))=\pi(f(R))$ for all permutations $\pi\colon A\to A$ and preference profiles $R$, $R'$.} Furthermore, the proof also reveals that a rank-based SCF that satisfies neutrality and strategyproofness can only choose a unique winner if this alternative is never uniquely bottom-ranked by a voter. 
\end{remark}

\begin{remark}
	\Cref{thm1} also holds under weaker versions of rank-basedness. First, our proof uses rank-basedness only in very specific situations, namely when two voters rename exactly two alternatives. Moreover, the only real restriction on the rank function $r$ is independence of the naming of other alternatives, i.e., $r(\succsim_i,a)=r(\succsim_i',a)$ for all preference relations~$\succsim_i$,~$\succsim_i'$ that only differ in the naming of alternatives in $A\setminus \{a\}$. Hence, we may also define rank-basedness based on a rank function other than the rank tuple and the result still holds. 
\end{remark}

\begin{remark}\label{rem:rankbstrict}
\Cref{thm1} does not hold when preferences are strict. For instance, the omninomination rule satisfies all required axioms for arbitrary numbers of voters and alternatives for strict preferences. It can even be shown that Claim 2 of the proof no longer holds for strict preferences as the following SCF is rank-based, Pareto-optimal, and strategyproof: if an alternative is top-ranked by every voter, this alternative is the unique winner; otherwise, return the alternatives which are top-ranked by the most and second most voters (in case of a tie return all alternatives with the second highest plurality score). However, no voter is a nominator for this rule. A proof of these claims and a formal definition of this SCF can be found in \Cref{prop:2PL} in the appendix. 
\end{remark}

\subsection{Support-Based SCFs}
It is not possible to replace rank-basedness with support-basedness in \Cref{thm1} since the Pareto rule is strategyproof, Pareto-optimal, and support-based. Note that the Pareto rule always chooses one of the most preferred alternatives of every voter. Consequently, Claim 1 in the proof of \Cref{thm1} cannot be true in general for support-based SCFs. Nevertheless, we show next that an analogue statement to Claim 2 remains true for such SCFs, i.e., every voter is a nominator for every support-based SCF that satisfies Pareto-optimality and strategyproofness. 

\begin{restatable}{theorem}{thmSB}\label{thm2}
In every support-based SCF that satisfies Pareto-optimality and strategyproofness, every voter is a nominator if $m\geq 3$.
\end{restatable} 
\begin{proof}
	Let $f$ be a support-based SCF satisfying Pareto-optimality and strategyproofness for fixed numbers of voters $n\geq 1$ and alternatives $m\geq 3$. For $n = 1$, the theorem follows immediately from Pareto-optimality as only the most preferred alternatives of the single voter are Pareto-optimal. Moreover, \Cref{lem:NU} proves the theorem for $n=2$. Indeed, if a voter is not a nominator, support-basedness shows that no voter is a nominator. Hence, \Cref{lem:NU} shows that every voter is a dictator, which means that $f(R)=\{a\}$ and $f(R)=\{b\}$ are simultaneously true if voter 1 prefers $a$ uniquely the most and voter 2 prefers $b$ uniquely the most. This is a contradiction and proves the theorem if $n=2$.
	
	Therefore, we focus on the case that $n\geq 3$ and assume for contradiction that a voter is not a nominator for $f$. We derive from this assumption by an induction on $k\in \{1,\dots, n-1\}$ that every set of $n-k$ voters is decisive. This results in a contradiction when $k\geq n/2$ because then, two alternatives can be simultaneously top-ranked by $n-k\leq n/2$ voters, and both of them must be the unique winner.
	
	The induction basis $k=1$ follows directly from \Cref{lem:NU}: support-basedness implies that if a single voter is not a nominator for $f$, no voter is a nominator for $f$ as we can just rename the voters. Hence, every set of size $n-1$ is decisive. Next, we assume that our claim holds for a fixed $k\in \{1,\dots, n-2\}$ and prove that also every set of $n-(k+1)$ voters is decisive. For this, we focus only on three alternatives $a,b,c$ and on a certain partition of the voters. This is possible as the induction hypothesis allows us to exchange the roles of the alternatives without affecting the proof and support-basedness allows us to reorder the voters. Thus, consider the profile $R^{k,1}$, in which $X=A\setminus \{a,b,c\}$, and note that $f(R^{k,1})=\{a\}$ because of \Cref{lem:NU}.\smallskip
	
	\begin{profile}{C{0.06\profilewidth}C{0.27\profilewidth} C{0.27\profilewidth}C{0.27\profilewidth}}
		{$R^{k,1}$:} & $[1\dots k]$: $a,X,c,b$  & $k+1$: $c,X,b,a$ & $[k+2\dots n]$: $a, b, X, c$ 
	\end{profile}

	Next, we aim to reverse the preferences of the voters $\seti{k+2}{n}$ over $a$ and $b$. This is achieved by the repeated application of the following steps explained for voter $k+2$. First, voter $k+2$ changes his preference to $\{a,b\}, c, X$ to derive the profile $R^{k,2}$. Since a subset of $\{a,b\}$ was chosen before this step, \ref{imp2} implies that $f(R^{k,2})\subseteq \{a,b\}$. Next, we use support-basedness to exchange the preferences of voter $k+1$ and $k+2$ over $a$ and $b$. This leads to the profile $R^{k,3}$ and support-basedness implies that $f(R^{k,3})=f(R^{k,2})\subseteq \{a,b\}$. Since $\{a,b\}= B_{k+1}(R^{k,3})$, \ref{imp1} implies that this voter cannot make another alternative win by manipulating. Thus, he can switch back to his original preference to derive $R^{k,4}$ and the fact that $f(R^{k,4})\subseteq \{a,b\}$. \smallskip
	
	\begin{profile}{C{0.06\profilewidth} C{0.2\profilewidth} C{0.2\profilewidth} C{0.2\profilewidth} C{0.24\profilewidth}}
		$R^{k,2}$: &$[1\dots k]$: $a,X,c,b$ & $k+1$: $c,X,b,a$ & $k+2$: $\{a,b\},X,c$ & $[k+3\dots n]$: $a, b, X,c$
	\end{profile}

	\begin{profile}{C{0.06\profilewidth} C{0.2\profilewidth} C{0.2\profilewidth} C{0.2\profilewidth} C{0.24\profilewidth}}
		$R^{k,3}$: &$[1\dots k]$: $a,X,c,b$ & $k+1$: $c,X,\{a,b\}$ & $k+2$: $b,a,X,c$ & $[k+3\dots n]$: $a, b, X,c$ 
	\end{profile}

	\begin{profile}{C{0.06\profilewidth} C{0.2\profilewidth} C{0.2\profilewidth} C{0.2\profilewidth} C{0.24\profilewidth}}
		$R^{k,4}$: &$[1\dots k]$: $a,X,c,b$ & $k+1$: $c,X,b,a$ & $k+2$: $b,a,X,c$ & $[k+3\dots n]$: $a, b, X,c$ \\
	\end{profile}
	
	It is easy to see that we can repeat these steps for every voter $\seti{k+3}{n}$. This process results in the profile $R^{k,5}$ and shows that $f(R^{k,5})\subseteq\{a,b\}$. Moreover, consider the profile $R^{k,6}$ derived from $R^{k,5}$ by letting voter $k+1$ make $b$ his best alternative. Because $n-k$ voters prefer $b$ uniquely the most in $R^{k,6}$, the induction hypothesis entails that $f(R^{k,6})=\{b\}$. This means that voter $k+1$ can manipulate by switching from $R^{k,5}$ to $R^{k,6}$ if $f(R^{k,5})=\{a\}$ or $f(R^{k,5})=\{a,b\}$. Consequently, $f(R^{k,5})=\{b\}$ is the only valid choice set for $R^{k,5}$.\smallskip
	
	\begin{profile}{C{0.06\profilewidth}C{0.27\profilewidth} C{0.27\profilewidth}C{0.27\profilewidth}}
	     $R^{k,5}$: &$[1\dots k]$: $a,X,c,b$ & $k+1$: $c,X,b,a$ & $[k+2\dots n]$: $b, a, X,c$ 
	\end{profile}

	\begin{profile}{C{0.06\profilewidth}C{0.27\profilewidth} C{0.27\profilewidth}C{0.27\profilewidth}}
	     $R^{k,6}$: & $[1\dots k]$: $a,X,c,b$ & $k+1$: $b,X,a,c$ & $[k+2\dots n]$: $b, a, X,c$ 
	\end{profile}

	So far, we have found a profile in which $b$ is uniquely chosen when the voters $\seti{k+2}{n}$ prefer it uniquely the most. Next, we show that this set of voters is therefore decisive. Hence, consider the profile $R^{k,7}$ which is derived from $R^{k,5}$ by letting the voters $\seti{1}{k}$ subsequently change their preference to $c,X,a,b$. Since $f(R^{k,5})=\{b\}$ and $b$ is the worst alternative for these voters, \ref{imp1} implies that $f(R^{k,7})=\{b\}$.
	As last step, we change the preferences of voter $k+1$ such that $b$ is his least preferred alternative. For this, we first let all voters $\seti{k+2}{n}$ subsequently change their preference to $b, X, c, a$. This modification results in the profile $R^{k,8}$ and \ref{imp2} implies that $f(R^{k,8})=\{b\}$. Moreover, observe that alternative $a$ is Pareto-dominated by $c$ in $R^{k,8}$. Therefore, voter $k+1$ can now swap $a$ and $b$ to derive the profile $R^{k,9}$ and Pareto-optimality implies that $a\not\in f(R^{k,9})$. Then, strategyproofness implies that $f(R^{k,9})=\{b\}$ as any other subset of $A\setminus \{a\}$ is a manipulation for voter $k+1$.\smallskip
	
	\begin{profile}{C{0.06\profilewidth}C{0.27\profilewidth} C{0.27\profilewidth}C{0.27\profilewidth}}
	     $R^{k,7}:$ &$[1\dots k]$: $c,X,a,b$ & $k+1$: $c,X,b,a$ & $[k+2\dots n]$: $b, a, X,c$ 
	\end{profile}
	
	\begin{profile}{C{0.06\profilewidth}C{0.27\profilewidth} C{0.27\profilewidth}C{0.27\profilewidth}}
	    $R^{k,8}:$ &$[1\dots k]$: $c,X,a,b$ & $k+1$: $c,X,b,a$ & $[k+2\dots n]$: $b, X, c, a$
	\end{profile}

	\begin{profile}{C{0.06\profilewidth}C{0.27\profilewidth} C{0.27\profilewidth}C{0.27\profilewidth}}
	     $R^{k,9}:$ &$[1\dots k]$: $c,X,a,b$ & $k+1$: $c,X,a,b$ & $[k+2\dots n]$: $b, X, c, a$ 
	\end{profile}
	
	Finally, observe that the voters $\seti{1}{k+1}$ can change their preferences in $R^{k,9}$ arbitrarily without affecting the choice set because of \ref{imp1}, and the voters $\seti{k+2}{n}$ can reorder all alternatives in $A\setminus \{b\}$ without affecting the choice set because of \ref{imp2}. Thus, $b$ is always the unique winner if all voters $\seti{k+2}{n}$ prefer $b$ uniquely the most. Moreover, interchanging the roles of alternatives in this proof shows that every alternative is chosen uniquely if it is uniquely top-ranked by all voters in $\set{k+2}{n}$. 
	
	Next, we show that this set of voters is decisive and consider an arbitrary profile $R$ such that ${\succsim_i}={\succsim_j}$ for all $i,j\in{\set{k+2}{n}}$. If these voters only report a single alternative $x$ as top choice, $f(R)=\{x\}$ follows from our previous analysis. Otherwise, we start at a profile $R'$ in which the voters $\seti{k+2}{n}$ uniquely top-rank an alternative $x\in T_i(R)$. Once again, our previous analysis shows that $f(R')=\{x\}$ and if we let the voters $\seti{k+2}{n}$ sequentially deviate to their preference relation in $R$, \ref{imp2} shows that $f(R)\subseteq T_i(R)$ for $\seti{k+2}{n}$. Hence, this set is indeed decisive. Since support-basedness allows us to reorder the voters to derive that every set of $n-(k+1)$ voters is decisive, the induction step is proven. As a consequence, every voter is a nominator for a support-based SCF that satisfies strategyproofness and Pareto-optimality.
\end{proof}

\Cref{thm2} shows that every support-based SCF that satisfies Pareto-optimality and strategyproofness chooses one of the most preferred alternatives of every voter. Since the Pareto rule indeed satisfies all these criteria, this is no impossibility but demonstrates a high degree of indecisiveness. However, we can turn this result into an impossibility by strengthening support-basedness. For instance, if we require pairwiseness instead of support-basedness, \Cref{thm2} turns into an impossibility since pairwiseness and Pareto-optimality rule out that every voter is a nominator. For seeing this, consider the following two profiles $R^1$ and $R^2$.\smallskip 
 
 \begin{profile}{C{0.1\profilewidth} C{0.25\profilewidth} C{0.25\profilewidth} C{0.25\profilewidth}}
 	$R^1$:	& $1$: $\{a,b\},X$ &  $2$: $\{a,b\},X$ & $[3\dots n]$: $a,b,X$ \\
 	$R^2$: & $1$: $b,a,X$ &  $2$: $a,b,X$ & $[3\dots n]$: $a,b,X$ \\
 \end{profile}
 
It can be verified that $f(R^2)=f(R^1)=\{a\}$ for every pairwise and Pareto-optimal SCF $f$, which shows that voter $1$ is not a nominator. Hence, \Cref{thm2} implies that there is no pairwise, strategyproof, and Pareto-optimal SCF.\footnote{This impossibility was first observed by \shortciteA[Theorem~2]{BSS19a}.} As a consequence of this result, it follows also that no majoritarian SCF can satisfy Pareto-optimality and strategyproofness. As the next corollary shows, this impossibility is even true if we weaken Pareto-optimality to non-imposition.

\begin{corollary}\label{thm:C1}
	There is no majoritarian SCF that satisfies non-imposition and strategyproofness if $m\geq 3$ and $n\geq 3$.
\end{corollary}
\begin{proof}
	Assume for contradiction that there is a majoritarian SCF $f$ that satisfies non-imposition and strategyproofness for $m\geq 3$ and $n\geq 3$. As a first step, we show that $f$ satisfies Condorcet-consistency. Hence, choose an arbitrary alternative $a$ and consider a profile $R$ such that $f(R)=\{a\}$; such a profile exists by non-imposition. Next, we let the voters $i\in N$ one after another report $a$ as their favorite alternative. For each step, \ref{imp2} shows that the choice set does not change, and thus, this process results in a profile $R'$ with $f(R')=\{a\}$ and $T_i(R')=a$ for all voters $i\in N$. Furthermore, we infer from \ref{imp2} also that all voters can reorder all alternatives in $A\setminus \{a\}$ in $R'$ without affecting the choice set. Since $f$ is majoritarian, this means that $a$ is the unique winner for all profiles in which $a$ is the Condorcet winner, i.e., $f$ is Condorcet-consistent. 
	
	As a consequence of this observation, every group $I$ of at least $\lceil \frac{n+1}{2}\rceil$ voters can enforce that $f$ chooses an alternative $x$ uniquely if all voters in $I$ report $x$ as their unique top choice. Moreover, an analogous argument as in Step 4 of \Cref{lem:NU} shows that every such group $I$ is decisive for $f$. Note that this statement is equivalent to the induction hypothesis in the proof of \Cref{thm2} when $k=\lfloor \frac{n-1}{2}\rfloor$. Since the first steps of this proof do not require Pareto-optimality, we derive that $f(R^1)=\{b\}$ because the profile $R^1$ corresponds to the profile $R^{k, 5}$ in the proof of \Cref{thm2}, where $k=\lfloor \frac{n-1}{2}\rfloor$.\smallskip
	
	\begin{profile}{C{0.1\profilewidth} C{0.25\profilewidth} C{0.25\profilewidth} C{0.25\profilewidth}}
		$R^{1}$: &$[1\dots k]$: $a,X,c,b$ & $k+1$: $c,X,b,a$ & $[k+2\dots n]$: $b, a, X,c$ 
	\end{profile}
	
	As next step, we let the voters $\seti{1}{k}$ deviate one after another by reporting $c,b,X,a$. Since $b$ is the least preferred alternative of these voters in $R^1$ and $f(R^1)=\{b\}$, \ref{imp1} requires that the choice set does not change during these steps. Hence, it holds for the resulting profile $R^2$ that $f(R^2)=\{b\}$. Furthermore, one after another, we let the voters $\seti{k+2}{n}$ make $c$ into their second best alternative. This results in the profile $R^3$ and \ref{imp2} implies that $f(R^3)=\{b\}$.\smallskip
	
	\begin{profile}{C{0.1\profilewidth} C{0.25\profilewidth} C{0.25\profilewidth} C{0.25\profilewidth}}
		$R^{2}$: &$[1\dots k]$: $c,b,X,a$ & $k+1$: $c,X,b,a$ & $[k+2\dots n]$: $b, a, X,c$ \\
		$R^{3}$: &$[1\dots k]$: $c,b,X,a$ & $k+1$: $c,X,b,a$ & $[k+2\dots n]$: $b, c,a,X$ 
	\end{profile}

	Finally, note that $f(R^3)=\{b\}$ is a contradiction. If $n$ is odd, then $c$ is the Condorcet winner in $f(R^3)$ and thus, Condorcet-consistency requires that $f(R^3)=\{c\}$. On the other hand, if $n$ is even, we can exchange the roles of $b$ and $c$ in the derivation of $R^3$ to derive that $f(R^3)=\{c\}$ must also be true. This is possible as $c\sim_{R^3} b$ and $x\succ_{R^3} y$ for all $x\in \{b,c\}$, $y\in A\setminus \{b,c\}$. Hence, if we exchange the role of $b$ and $c$ in the derivation of the profile $R^3$, we end up with another profile $R^{3'}$ with the same majority relation, but our proof shows that $f(R^{3'})=\{c\}$. This is in conflict with $f$ being majoritarian and thus, no majoritarian SCF satisfies both strategyproofness and non-imposition if $m\geq 3$ and $n\geq 3$. 
\end{proof}

\begin{remark}\label{remark}
All axioms used in \Cref{thm2} are required as the following SCFs show.
Every constant SCF satisfies support-basedness and strategyproofness, and violates Pareto-optimality and that every voter is a nominator. 
The SCF that always chooses a unique Pareto-optimal alternative according to a fixed tie-breaking order satisfies Pareto-optimality and support-basedness but violates strategyproofness and that every voter is a nominator. An SCF that satisfies Pareto-optimality and strategyproofness but violates support-basedness and that every voter is a nominator can be found as follows. We define a transitive dominance relation by slightly strengthening Pareto-dominance. Therefore, we additionally allow that an alternative~$a$ that is among the most preferred alternatives of $n-1$ voters can dominate another alternative $b$, even if a single voter strictly prefers $b$ to $a$. More formally, we say that an alternative~$a$ \emph{dominates} alternative $b$ if $a$ Pareto-dominates $b$ or $n-1$ voters prefer $a$ the most while $s_{ab}(R)\geq 2$ and $s_{ba}(R)\leq 1$. It should be stressed that it is not required that $a$ is uniquely top-ranked by $n-1$ voters, but only that it is among their best alternatives. The SCF $f^*$ that chooses all maximal elements with respect to this dominance relation satisfies all required properties (see \Cref{prop:f*} in the appendix for more details).
Also the bound on $m$ is tight as the majority rule
satisfies all axioms if $m=2$ but no voter is a nominator for this SCF.
\end{remark}

\begin{remark}\label{rem:SD}
\Cref{thm2} implies an impossibility for $m\geq 4$ and $n\geq 4$ if we strengthen Pareto-optimality to \sd-efficiency (also known as ordinal efficiency, see \shortciteR{BoMo01a}). This result follows by considering the preference profiles $R^1$ and $R^2$ shown below. In this profile, $X=A\setminus \{a,b,c,d\}$ and all voters $\seti{5}{n}$ are assumed to be indifferent between all alternatives.\smallskip

\begin{profile}{C{0.06\profilewidth} C{0.2\profilewidth} C{0.2\profilewidth} C{0.2\profilewidth} C{0.2\profilewidth}}
	$R^1$: & $1$: $a,c,b,d,X$ &  $2$: $a,d,b,c,X$ & $3$: $b,c,a,d,X$ & $4$: $b,d,a,c,X$ \\
	$R^2$: & $1$: $a,b,c,d,X$ &  $2$: $a,b,d,c,X$ & $3$: $c,b,a,d,X$ & $4$: $d,b,a,c,X$ \\
\end{profile}

Next, consider a support-based and \sd-efficient SCF $f$. \sd-efficiency implies that $f(R^1)\cap X=\emptyset$ and that either $c\not\in f(R^1)$ or $d\not\in f(R^1)$. Moreover, support-basedness implies that $f(R^1)=f(R^2)$, which means that voter 3 or 4 is not a nominator for $f$. Thus, \Cref{thm2} implies that $f$ is not strategyproof, which shows the incompatibility of strategyproofness, support-basedness, and \sd-efficiency. Note that there is also no rank-based, strategyproof, and \sd-efficient SCF if $m\geq 4$ and $n\geq 3$ because of \Cref{thm1}. This leads to the important and challenging open question whether there is an anonymous, SD-efficient, and strategyproof SCF. 
\end{remark}

\begin{remark}
Just as in the proof of \Cref{thm1}, we make only very restricted use of support-basedness in the proof of \Cref{thm2}. It suffices if two voters are allowed to exchange their preferences over two alternatives. This technical restriction is significantly weaker than support-basedness, which allows any number of voters to change their preferences. 
\end{remark}

\begin{remark}
If preferences are required to be strict, \Cref{thm2} does not hold. Several SCFs including the uncovered set, the minimal covering set, and the essential set are strategyproof, Pareto-optimal and support-based, but no voter is a nominator for them \shortcite<for more details, see, e.g.,>[Chapter 3]{BCE+14a}. Remarkably, all these SCFs are majoritarian and thus affected by the stronger impossibility in \Cref{thm:C1} if we allow for ties in voters' preferences. 
\end{remark}

\begin{remark}
\Cref{thm1,thm2} raise the question whether all voters must be nominators for every anonymous, Pareto-optimal, and strategyproof social choice function. This is not the case because the SCF $f^*$, as defined in \Cref{remark}, satisfies near unanimity and therefore represents a counterexample. This leads to the intriguing question on the minimal value $l$ such that all groups of $l$ voters are nominating for every anonymous SCF that satisfies Pareto-optimality and strategyproofness. An upper bound for this problem is provided by \Cref{lem:NU}, which shows that $l\leq \lceil\frac{n+1}{2}\rceil$.
\end{remark}

\subsection{Non-Imposing SCFs}

Finally, we consider the class of non-imposing SCFs. Recall that an SCF is non-imposing if every alternative is returned as the unique winner in some profile. Among the SCFs typically studied in social choice theory, there are only very few that fail to be non-imposing, e.g., SCFs that never return certain alternatives (such as constant SCFs) or SCFs that never return singletons.
We will show a rather strong consequence of strategyproofness for non-imposing SCFs: every such function has to return the Condorcet loser in at least one preference profile and thus violates the Condorcet loser property. 
In the presence of neutrality, non-imposition can be seen as a decisiveness requirement. More precisely, if a neutral SCF fails non-imposition, it can never return a singleton choice set, which means that the SCF has to choose unreasonably large choice sets for many preference profiles. For instance, even if all voters agree on a best alternative, it cannot be chosen uniquely. Accordingly, the theorem identifies a tradeoff between decisiveness and the undesirable property of selecting Condorcet losers.

Similar to the proofs of \Cref{thm1} and \Cref{thm2}, we start with a general lemma on strategyproof SCFs. This time, we investigate the relationship between vetoing and decisive groups of voters and show that these notions coincide for strategyproof and non-imposing SCFs.

\begin{restatable}{lemma}{lemCL}\label{lem:CL}
	Let $f$ denote a strategyproof and non-imposing SCF for $m\geq 2$ alternatives. A group of voters $I$ with $\emptyset\subsetneq I\subseteq N$ is vetoing for $f$ if and only if it is decisive for $f$.  
\end{restatable}
\begin{proof}
	Let $f$ denote a strategyproof and non-imposing SCF and consider a non-empty set of voters $I \subseteq N$. First, observe that if $I$ is decisive for $f$, then it is also vetoing: if all voters in $I$ agree on the same preference relation in a profile $R$ and report an alternative $x$ as their unique last choice, then decisiveness requires that $f(R)\subseteq T_i(R)$ for $i\in I$. This means that $x\not\in f(R)$ because $x\not\in T_i(R)$, which shows that the group $I$ is also vetoing for $f$. 
	
	Next, suppose that $I$ is vetoing for $f$. If $m=2$, every vetoing group of voters is decisive because such a group can determine the choice set by vetoing out an alternative. Thus, we focus in the sequel on the case $m\geq 3$ and show that the set $I$ is decisive for $f$. We derive this claim in three steps: first, we show that $f(R)=\{x\}$ if all voters report $x$ as their unique choice. Next, we prove that $f(R)=\{x\}$ is also true if only the voters $i\in I$ report $x$ as their best choice. Finally, we infer from this insight that the set $I$ is decisive for $f$.\bigskip
	
	\textbf{Step 1:}
	As a first step, we show that $f(R)=\{x\}$ for all alternatives $x\in A$ and preference profiles $R$ such that all voters report $x$ as unique top choice. For proving this claim, consider an arbitrary alternative $a$ and a preference profile $R^1$ such that all voters uniquely top-rank $a$ in $R^1$. Since $a$ and $R^1$ are chosen arbitrarily, the claim follows by showing that $f(R^1)=\{a\}$. Note for this that there is a profile $R^2$ such that $f(R^2)=\{a\}$ because $f$ is non-imposing. As next step, we derive $R^1$ from $R^2$ by sequentially replacing the preference relation of all voters $i\in N$ with $\succsim_i^1$. In more detail, consider the sequence $R^{2,0},\dots, R^{2,n}$ such that $R^{2,0}=R^2$, $R^{2,n}=R^1$, and $R^{2,i}$ evolves out of $R^{2,i-1}$ by replacing $\succsim_i^{2}$ with $\succsim_i^1$. Now, if $f(R^{2,i-1})=\{a\}$ for some $i\in \{1,\dots, n\}$, \ref{imp2} implies that $f(R^{2,i})=\{a\}$ because $T_i(R^{2,i})=\{a\}$. Since $f(R^2)=\{a\}$, we can repeatedly use this argument to derive that $f(R^1)=\{a\}$, which proves this step.\bigskip
	
	\textbf{Step 2:}
	Our next goal is to show that $f(R)=\{x\}$ for all alternatives $x\in A$ and profiles $R$ such that all voters in $I$ report $x$ as their unique top choice. Hence, we suppose that $I\subsetneq N$; otherwise, we can directly proceed with Step 3. Just as in the last step, consider an arbitrary alternative $a$. Subsequently, we will show that $f(R^3)=\{a\}$ for the following profile $R^3$.\smallskip
	
	\begin{profile}{C{0.1\profilewidth} C{0.35\profilewidth} C{0.4\profilewidth}}
		{$R^3$:}& $I$: $a,A\setminus\{a\}$ & $N\setminus I$: $A\setminus \{a\}, a$  
	\end{profile}
	
	This insight suffices to prove Step 2 since \ref{imp2} allows the voters in $I$ to reorder all alternatives in $A\setminus \{a\}$ arbitrarily without affecting the choice set, and \ref{imp1} allows the voters in $N\setminus I$ to reorder all alternatives without affecting the choice set.
	
	For proving the claim on $R^3$, let $A=\{a, a_1, \dots, a_{m-1}\}$ denote the set of alternatives and, given $l\in \{1, \dots, m-1\}$ and a set of alternatives $X\subseteq A\setminus \{a, a_l\}$, define the profiles $R^{l,X}$ as shown below.\smallskip
	
	\begin{profile}{C{0.1\profilewidth} C{0.35\profilewidth} C{0.4\profilewidth}}
		{$R^{l,X}$:}& $I$: $a,A\setminus\{a\}$ & $N\setminus I$: $a_l, \{a\}\cup X, A\setminus (X\cup\{a, a_l\})$  
	\end{profile}
	
	\noindent In particular, the preference profiles $R^{l,\emptyset}$ and $R^{l, A\setminus\{a,a_l\}}$ are defined as follows.\smallskip
	
	\begin{profile}{C{0.12\profilewidth} C{0.35\profilewidth} C{0.4\profilewidth}}
		{$R^{l,\emptyset}$:}& $I$: $a,A\setminus\{a\}$ & $N\setminus I$: $a_l, a, A\setminus \{a, a_l\}$\\
		{$R^{l,A\setminus \{a, a_l\}}$:}& $I$: $a,A\setminus\{a\}$ & $N\setminus I$: $a_l, A\setminus \{a_l\}$
	\end{profile}
		
	Note that strategyproofness implies that, if $f(R^{l, A\setminus \{a, a_l\}})=\{a\}$ for all $l\in \{1,\dots, m-1\}$, then $f(R^3)=\{a\}$. The reason for this is the following: starting at an arbitrary profile $R^{l, A\setminus \{a, a_l\}}$, we can derive the profile $R^3$ by letting the voters $i\in N\setminus I$ change their preference relation one after another to $\succsim_i^3$. For each step, \ref{imp1} implies that a subset of $A\setminus \{a_l\}$ needs to be chosen as otherwise, the deviation is a manipulation. Hence, we derive for every $l\in \{1,\dots, m-1\}$ that $a_l\not\in f(R^3)$, which means that $f(R^3)=\{a\}$. 
	
	Subsequently, we will prove by induction on $z=|X|$ that $f(R^{l,X})=\{a\}$ for all $l\in \{1,\dots, m-1\}$ and $X\subseteq A\setminus \{a, a_l\}$. Because of our previous insights, this completes the proof of Step~2. Two observations are central for the subsequent argument: firstly, our argument is closed under renaming alternatives in $A\setminus \{a\}$. This means that if we can show that $f(R^{l,X})=\{a\}$ for some $l\in \{1, \dots, m-1\}$ and $X\subseteq A\setminus \{a, a_l\}$ with $|X|=z$, this result holds for all $l'\in \{1, \dots, m-1\}$ and subsets of $A\setminus \{a, a_{l'}\}$ with size $z$. Secondly, we can ensure that any alternative $a_l\in \{a_1, \dots, a_{m-1}\}$ is unchosen by letting the voters $i\in I$ report it as their unique bottom choice. Furthermore, this step does not change the choice set if $a$ is the unique winner because of \ref{imp2}.
	
	First, we prove the base case $z=0$, i.e., we show that $f(R^{l,\emptyset})=\{a\}$ for all $l\in \{1,\dots, m-1\}$. Consider for this the following profiles and note that we display again $R^{l,\emptyset}$ so that all relevant profiles are shown.\smallskip
	
	\begin{profile}{C{0.1\profilewidth} C{0.35\profilewidth} C{0.4\profilewidth}}
		{$\hat{R}^{l,\emptyset}$:}& $I$: $a,A\setminus\{a,a_l\},a_l$ & $N\setminus I$: $a, a_l, A\setminus \{a, a_l\}$\\
		{$\tilde{R}^{l,\emptyset}$:}& $I$: $a,A\setminus\{a,a_l\},a_l$ & $N\setminus I$: $a_l, a, A\setminus \{a, a_l\}$\\
		${R}^{l,\emptyset}$:& $I$: $a,A\setminus\{a\}$ & $N\setminus I$: $a_l, a, A\setminus \{a, a_l\}$
	\end{profile}
	
	Recall that, by Step~1, $f(R)=\{a\}$ if all voters $i\in N$ uniquely top-rank~$a$. Consequently, it holds that $f(\hat{R}^{l,\emptyset})=\{a\}$. Next, observe that $a_l$ is the uniquely least preferred alternative of the voters in $I$. Now, let the voters in $N\setminus I$ swap $a$ and $a_l$ one after another to derive $\tilde{R}^{l,\emptyset}$. For each step, $a_l$ cannot be chosen because the voters in $I$ veto it. Hence, if $a$ is chosen before the deviation of a voter $i\in N\setminus I$, it follows that $a$ needs to be chosen afterwards; otherwise, a set $X$ with $\{a\}\subsetneq X\subseteq A\setminus \{a_l\}$ is chosen and thus, voter $i$ can manipulate by reverting this modification. Since $f(\hat R^{l,\emptyset})=\{a\}$, we infer therefore that $f(\tilde{R}^{l,\emptyset})=\{a\}$. 	 Finally, note that $\tilde{R}^{l, \emptyset}$ only differs from $R^{l,\emptyset}$ in the preferences of the voters in $I$ on the alternatives $A\setminus \{a\}$. Since \ref{imp2} allows us to reorder the preferences of these voters on $A\setminus \{a\}$ arbitrarily without affecting the choice set, we derive that $f(R^{l,\emptyset})=\{a\}$ for all $l\in \{1, \dots, m-1\}$.

	Next, we focus on the induction step, i.e., we assume that $f(R^{l,X})=\{a\}$ for all $l\in \{1,\dots, m-1\}$ and all $X\subseteq A\setminus \{a, a_j\}$ with $|X|=z-1$ and show that the same is true for all sets $X'$ of size $z$. Recall for this that the derivation of $f(R^{l,X'})=\{a\}$ is independent of the naming of the alternatives in $A\setminus \{a\}$, and thus, it suffices to show that $f(R^{z+1, \{a_1, \dots, a_z\}})=\{a\}$. For this, let $Z=\{a_1, \dots, a_z\}$, $Z_{+a}=Z\cup\{a\}$, and $Z_{-l}=Z\setminus \{a_l\}$ for every $l\in \{1, \dots, z\}$, and consider the following profiles, where $l\in \{1,\dots, z\}$.\smallskip
	
	\begin{profile}{C{0.1\profilewidth} C{0.35\profilewidth} C{0.4\profilewidth}}
		{$R^{l, Z_{-l}}$: }& $I$: $a,A\setminus\{a\}$ & $N\setminus I$: $a_l, Z_{-l}\cup\{a\}, A\setminus Z_{+a}$\\
		{$\hat{R}^{l,Z_{-l}}$: }& $I$: $a,A\setminus\{a, a_{z+1}\}, a_{z+1}$ & $N\setminus I$: $a_l, Z_{-l}\cup\{a\}, A\setminus Z_{+a}$\\
		{$\tilde{R}^{z+1, Z}$: }& $I$: $a,A\setminus\{a, a_{z+1}\}, a_{z+1}$ & $N\setminus I$: $a_{z+1}, Z_{+a}, A\setminus (Z\cup\{a,a_{z+1}\})$\\
		{${R}^{z+1, Z}$: }& $I$: $a,A\setminus\{a\}$ & $N\setminus I$: $a_{z+1}, Z_{+a}, A\setminus (Z\cup\{a,a_{z+1}\})$
	\end{profile}
	
	Now, consider an arbitrary $l\in \{1,\dots, z\}$ and note that our induction hypothesis implies that $f(R^{l, Z_{-l}})=\{a\}$ since $|Z_{-l}|=z-1$. We derive $\hat{R}^{l, Z_{-l}}$ from $R^{l, Z_{-l}}$ by letting the voters $i\in I$ sequentially change their preference relations such that $a_{z+1}$ is their uniquely least preferred alternative. For each step, \ref{imp2} shows that the choice set is not allowed to change and thus, we infer that $f(\hat{R}^{l, Z_{-l}})=\{a\}$. In particular, since all voters in $I$ report $a_{z+1}$ as their unique least preferred alternative, it follows that $a_{z+1}\not\in f(\hat{R}^{l, Z_{-l}})$ regardless of the preference relations of the voters in $N\setminus I$. 
	
	Hence, we derive the profile $\tilde{R}^{z+1, Z}$ from $\hat{R}^{l,Z_{-l}}$ by replacing the preference relations of the voters in $N\setminus I$ one after another with $a_{z+1}, Z_{+a}, A\setminus (Z_{+a}\cup\{a_{z+1}\})$. Next, we will investigate one of these steps in detail and thus, let $R^l$ and $\bar R^l$ denote two consecutive profiles in the derivation of $\tilde{R}^{z+1, Z}$ from $\hat{R}^{l, Z_{-l}}$. Moreover, let $i$ denote the voter whose preference relation is different in $R^l$ and $\bar R^l$. First, note that $a_{z+1}\not\in f(R^l)$ and $a_{z+1}\not\in f(\bar R^l)$ because this alternative is vetoed out. Thus, if $f(R^l)\subseteq Z_{+a}\setminus \{a_l\}$, then $f(\bar R^l)\subseteq Z_{+a}$ because voter $i$ can manipulate $f$ by switching from $\bar R^l$ to $R^l$ otherwise. Additionally, $a_l$ cannot be chosen in $\bar R^l$; otherwise, voter~$i$ can manipulate $f$ by deviating from $R^l$ to $\bar R^l$. This means that if $f(R^l)\subseteq Z_{+a}\setminus \{a_l\}$, then $f(\bar R^l)\subseteq Z_{+a}\setminus \{a_l\}$. Since $f(\hat{R}^{l, Z_{-l}})=\{a\}$, it follows from a repeated application of this argument that $f(\tilde{R}^{z+1, Z})\subseteq Z_{+a}\setminus \{a_l\}$. Finally, note that we can apply this argument for every $l\in \{1,\dots,z\}$. This entails that $f(\tilde{R}^{z+1})\subseteq Z_{+a}\setminus Z$, i.e., $f(\tilde{R}^{z+1})=\{a\}$.
	
	As last step, we derive ${R}^{z+1, Z}$ from $\tilde{R}^{z+1, Z}$ by reordering the preference relations of the voters $i\in I$. Because $a$ stays their best alternative, it follows from a repeated application of \ref{imp2} that $f(R^{z+1,Z})=\{a\}$. Since the argument is closed under renaming alternatives in $A\setminus \{a\}$, this proves the induction step.\bigskip
	
	\textbf{Step 3:}
	It remains to show that the set of voters $I$ is indeed decisive. Hence, consider an arbitrary profile $R$ in which all voters $i\in I$ report the same preference relation. We need to show that $f(R)\subseteq T_i(R)$ for all $i\in I$. For this, let $a$ denote an alternative in $T_i(R)$ for some voter $i\in I$, and let $R'$ denote a profile such that all voters in $I$ report $a$ as their uniquely best alternative, and all voters in $N\setminus I$ report the same preference relation as in $R$. By Step 2, it follows that $f(R')=\{a\}$. Next, we let the voters $i\in I$ revert one after another back to $\succsim_i$. Since $a\in T_i(R)$ and ${\succsim_i}={\succsim_j}$ for all $i,j\in I$, it follows from a repeated application of \ref{imp2} that $f(R)\subseteq T_i(R)$. This proves that every vetoing group is also decisive for $f$. 
\end{proof}

\Cref{lem:CL} has several interesting consequences. First of all, it shows that the notions of vetoing and decisive groups are equivalent for strategyproof SCFs that satisfy non-imposition. Consequently, no strategyproof, non-imposing, and non-dictatorial SCF can have a vetoer. Furthermore, our lemma entails that there cannot be two disjoint vetoing groups of voters for such SCFs. The reason for this is that both sets need to be decisive for such an SCF, but there cannot be two disjoint decisive groups. In particular, this means that no group of voters $I$ with $|I|\leq\frac{n}{2}$ can be vetoing for an anonymous, strategyproof, and non-imposing SCF.

Furthermore, \Cref{lem:CL} shows that every group of voters $I$ with $|I|>\frac{n}{2}$ is decisive for a strategyproof SCF that satisfies non-imposition and the Condorcet loser property. The reason for this is that the Condorcet loser property entails that such groups are vetoing. Next, we use this insight to show that there is no strategyproof SCF that satisfies the Condorcet loser property and non-imposition. Note that we present here a simplified proof with completely indifferent voters. In the appendix, we give a more involved proof which avoids such artificial voters. 

\begin{restatable}{theorem}{CLandNI}\label{thm4}
	There is no strategyproof SCF that satisfies the Condorcet loser property and non-imposition if $m\geq 3$ and $n\geq 4$. 
\end{restatable}
\begin{proof}
	We prove the statement by induction over $n\ge 4$.\medskip
	
	\textbf{Induction basis:} Assume for contradiction that there is a strategyproof SCF $f$ for $n=4$ voters and $m\geq 3$ alternatives that satisfies the Condorcet loser property and non-imposition. First, consider the profiles $R^1$ to $R^5$ shown below.\smallskip
	
	\begin{profile}{C{0.05\profilewidth} C{0.21\profilewidth} C{0.21\profilewidth} C{0.21\profilewidth} C{0.21\profilewidth}}
		$R^1$: & $1$: $a, c, X, b$    & $2$: $a, b, X, c$ & $3$: $a, b, X, c$ & $4$: $b, X, c, a$\\
		$R^2$: & $1$: $\{a,c\}, X+b$  & $2$: $a, b, X, c$ & $3$: $a, b, X, c$ & $4$: $b, X, c, a$\\
		$R^3$: & $1$: $\{a,c\}, X+b$  & $2$: $a, c, X, b$ & $3$: $a, b, X, c$ & $4$: $c, X, \{a,b\}$\\
		$R^4$: & $1$: $\{a,c\}, X+b$  & $2$: $a,c, X, b$ & $3$: $b,a, X, c$ & $4$: $c,X,\{a,b\}$\\
		$R^5$: & $1$: $\{a,c\}, X+b$  & $2$: $a,c, X, b$ & $3$: $b,a, X, c$ & $4$: $c,b,X,a$
	\end{profile}
	
	\Cref{lem:CL} shows that $f(R^1)=\{a\}$ since every group of $3$ voters needs to be decisive for $f$. Moreover, $c$ is the Condorcet loser in $R^1$, even if voter 1 is indifferent between $a$ and $c$. Thus, we replace next the preference relation of voter 1 with $\{a,c\}, X+b$, where $X+b=X\cup \{b\}$, to derive the profile $R^2$. \ref{imp2} implies that $f(R^2)\subseteq \{a,c\}$ because otherwise voter 1 can manipulate by reverting back to $R^1$. Moreover, $c\not\in f(R^2)$ due to the Condorcet loser property and we hence deduce that $f(R^2)=\{a\}$. As second step, we let voter 2 change his preference relation to $a,c,X,b$ and voter~4 change his preference to $c,X,\{a,b\}$ in order to make $b$ into the Condorcet loser. It follows from \ref{imp1} and \ref{imp2} that the choice set does not change during these steps since the current winner $a$ is the best alternative of voter~2 after the manipulation and the worst alternative of voter~$4$ before the manipulation. Hence, these steps result in the profile $R^3$ with $f(R^3)=\{a\}$. Furthermore, observe that $b$ is the Condorcet loser in $R^3$, even if voter 3 swaps $a$ and $b$. Thus, we derive the profile $R^4$ from $R^3$ by applying this modification. The Condorcet loser property requires that $b\not\in f(R^4)$, which entails, in turn, that $f(R^4)=\{a\}$; otherwise, voter~3 can manipulate $f$ by deviating from $R^4$ to $R^3$. As last step, we let voter 4 change his preference relation to $c,b,X,a$ to derive the profile $R^5$. Since $f(R^4)=\{a\}\subseteq B_{4}(R^4)$, it follows from \ref{imp1} that $f(R^5)\subseteq \{a,b\}$.
	
	Next, observe that we can apply analogous steps for profiles that are symmetric with respect to the voters or alternatives. Thus, we infer for the choice sets of the profiles $R^6$, $R^7$, and $R^8$ that $f(R^6)\subseteq \{a,c\}$, $f(R^7)\subseteq \{a,b\}$, and $f(R^8)\subseteq \{b,c\}$.\smallskip
	
	\begin{profile}{C{0.05\profilewidth} C{0.21\profilewidth} C{0.21\profilewidth} C{0.21\profilewidth} C{0.21\profilewidth}}
		$R^6$: & $1$: $\{b,c\}, X+a$  & $2$: $a,c, X, b$ & $3$: $b,a, X, c$ & $4$: $c,b,X,a$\\
		$R^7$: & $1$: $a,b, X, c$ & $2$: $c,a, X, b$  & $3$: $\{b,c\}, X+a$ & $4$: $b,c,X,a$\\
		$R^8$: & $1$: $a,b, X, c$ & $2$: $c,a, X, b$  & $3$: $\{a,c\}, X+b$ & $4$: $b,c,X,a$
	\end{profile}
	
	Note that, if $b\in f(R^5)$, then voter 1 can manipulate by switching to $R^6$ as $f(R^6)\subseteq \{a,c\}$. Hence, we derive that $f(R^5)=\{a\}$. By a symmetric argument for $R^7$ and $R^8$, it follows that $f(R^7)=\{b\}$. 
	
	Finally, consider the profile $R^9$ shown below. \smallskip
	
	\begin{profile}{C{0.05\profilewidth} C{0.21\profilewidth} C{0.21\profilewidth} C{0.21\profilewidth} C{0.21\profilewidth}}
		$R^9$: & $1$: $a,b,X,c$  & $2$: $a,b,X,c,$ & $3$: $b,a, X, c$ & $4$: $b,a,X,c$
	\end{profile}
	
	We can derive the profile $R^9$ from $R^5$ and $R^7$. In more detail, we obtain $R^9$ from $R^5$ by replacing the preference relations of voters 1 and 2 with $a,b,X,c$ and the preference relation of voter $4$ with $b,a,X,c$. If we apply these steps one after another, \ref{imp1} and \ref{imp2} imply that $f(R^9)=\{a\}$. On the other hand, we obtain $R^9$ from $R^7$ by replacing the preference relation of voters 3 and 4 with $b,a,X,c$ and the preference relation of voter 2 with $a,b,X,c$ and obtain $f(R^9)=\{b\}$ by an analogous argument. This is a contradiction since $f(R^9)=\{a\}$ and $f(R^9)=\{b\}$ cannot be simultaneously true, which shows that there is no strategyproof SCF that satisfies non-imposition and the Condorcet loser property if $n= 4$ and $m\geq 3$.\medskip
	
	\textbf{Induction step:} Assume for contradiction that there is a strategyproof SCF $f$ for $n>4$ voters and $m\geq 3$ alternatives that satisfies non-imposition and the Condorcet loser property. Consider the following SCF $g$ for $n-1$ voters and $m$ alternatives: given a profile $R$ on $n-1$ voters, $g$ adds a new voter who is indifferent between all alternatives to derive a profile $R'$ on $n$ voters and returns $g(R) = f(R')$. Clearly, $g$ is strategyproof and inherits the Condorcet loser property from $f$. Furthermore, \Cref{lem:CL} shows that $g$ is non-imposing because every set of $n-1>\frac{n}{2}$ voters is decisive for $f$. Hence, we can construct a strategyproof SCF for $n-1$ voters that satisfies the Condorcet loser property and non-imposition if there is such an SCF for $n$ voters. Since our induction hypothesis states that no such SCF exists, we derive from the contraposition of this implication that there is no SCF satisfying all required axioms for $n> 4$ voters.	
\end{proof}

\begin{remark}\label{rem:indepCL}
	The axioms used in \Cref{thm4} are independent of each other. An SCF that only violates the Condorcet loser property is the Pareto rule. The SCF that returns all alternatives except the Condorcet loser only violates non-imposition. The SCF that returns all Pareto-optimal alternatives except the Condorcet loser only violates strategyproofness. The bounds on $n$ and $m$ are also tight. The majority rule satisfies all axioms if $m=2$, the Pareto rule satisfies all axioms if $n\leq2$, and a computer-aided approach proves that there is an SCF that satisfies all required axioms if $n=3$ and $m=3$. It is also possible to extend the SCF by hand to $m>3$ alternatives, but the resulting SCFs are purely technical and we therefore do not define them explicitly.
\end{remark}

\begin{remark}
	\shortciteA[Theorem~2]{Bran11c} has shown that no Condorcet extension can be strategyproof if $m\geq 3$ and $n\geq 3m$.
	By replacing the Condorcet loser property and non-imposition with Condorcet-consistency, careful inspection of the proof of \Cref{thm4} reveals that Condorcet-consistency and strategyproofness are already incompatible if $m\geq 3$ and $n\geq 4$. In particular, observe for this that $a$ is the Condorcet winner in the profile $R^4$ and thus, every Condorcet-consistent SCF satisfies $f(R^4)=\{a\}$. Departing from this insight, we can apply the same steps as in the proof of \Cref{thm4} since the Condorcet loser property is not used anymore.
\end{remark}

\begin{remark}
	\Cref{lem:CL} and \Cref{thm4} do not require the full power of non-imposition. For \Cref{lem:CL}, the following weakening holds: if an alternative $x$ can be uniquely chosen by a strategyproof SCF, then every vetoing group of voters $I$ can ensure that $x$ is the unique winner if they unanimously report it as their best alternative. For \Cref{thm4}, we can weaken non-imposition to the requirement that at least three alternatives can be returned as unique winner.
\end{remark}

\begin{remark}
A desirable strengthening of \Cref{thm4} would be to weaken the Condorcet loser property by only demanding that an alternative that is uniquely bottom-ranked by a majority of voters should not be chosen. A computer analysis has shown that this property is compatible with non-imposition and strategyproofness when $m\leq 3$ and $n\leq 6$, even if we additionally impose anonymity. We nevertheless believe that there may be an impossibility for larger values of $m$ and $n$. 
\todo{PL: anonymity and neutrality should be for free here; we can use an analogous averaging construction as in relax by using set union instead of probability averaging.}
\end{remark}

\section{Consequences for Randomized Social Choice}\label{sec:RSC}

So far, we have discussed our theorems in the context of set-valued social choice, but they also have consequences for randomized social choice, which is concerned with the study of \emph{social decision schemes (SDSs)}, i.e., functions that map preference profiles to \emph{lotteries} (i.e., probability distributions) over the alternatives. Since the notions of rank-basedness and support-basedness are independent of the type of the output of the function and merely define an equivalence relation over preference profiles, they can be straightforwardly extended to SDSs. For our other axioms, we consider variants in randomized social choice based on the support of lotteries, i.e., the set of alternatives with positive probability. For example, Pareto-optimality and the Condorcet loser property require that Pareto-dominated alternatives and Condorcet losers are always assigned probability~$0$. In other words, Pareto-optimality demands that Pareto-dominated alternatives are not in the support of any chosen lottery, a condition that is usually referred to as \emph{ex post efficiency}. Similarly, an SDS satisfies the Condorcet loser property if the Condorcet loser is never in the support of any chosen lottery. Next, an SDS satisfies non-imposition if every alternative is chosen with probability~$1$ for some profile. 
Finally, Kelly-strategyproofness translates to the notion of \dd-strategyproofness \shortcite{Bran17a}. To this end, we say that, given a preference relation~$\succsim$, a lottery $p$ \emph{deterministically dominates} a lottery $q$ if and only if $\supp(p) \succsim \supp(q)$. 
Then, an SDS $f$ is called \emph{\dd-strategyproof} if $f(R')$ 
does not strictly deterministically dominate 
$f(R)$ for all voters $i\in N$ and all profiles $R,R'$ such that ${\succsim_j}={\succsim_j'}$ for all $j\in N\setminus \{i\}$. 
Note that \dd-strategyproofness is weaker than most strategyproofness notions considered in the literature \shortcite<see, e.g.,>{Gibb77a,Bran17a,ABBB15a,BBEG16a}. In particular, it is weaker than weak \sd-strategyproofness as used by \shortciteA{BBEG16a} to prove a rather sweeping impossibility: no anonymous and neutral SDS is weakly \sd-strategyproof and \sd-efficient if $n\geq 4$ and $m\geq 4$. 

Based on these axioms for SDSs, we can translate our results to the randomized context. Note for this that we can turn every SDS $f$ into an SCF $g$ by returning the support of $f(R)$ instead of the lottery itself. Moreover, it is easy to verify that all our axioms carry over from the SDS $f$ to the SCF $g$ because they are only defined based on the support. Therefore, we derive the following corollaries. 

\begin{corollary}\label{col1}
	There is no rank-based SDS that satisfies \emph{ex post} efficiency and \dd-strategy\-proof\-ness if $m\geq 4$ and $n\geq 3$, or if $m\geq 5$ and $n\ge 2$.
\end{corollary}

\begin{corollary}\label{col2}
	Every support-based SDS that satisfies \emph{ex post} efficiency and \dd-strategy\-proof\-ness assigns positive probability to at least one most preferred alternative of every voter if $m\geq 3$.
\end{corollary}

\begin{corollary}\label{col3}
	There is no SDS that satisfies the Condorcet loser property, non-imposition, and \dd-strategyproofness if $m\geq 3$ and $n\geq 4$.
\end{corollary}

\Cref{col1} can be seen as a strengthening of the impossibility of \shortciteA{BBEG16a} for the class of rank-based SDSs as we require both a weaker strategyproofness notion and a weaker efficiency notion. However, our result only holds for rank-based SDSs rather than the more general class of anonymous SDSs. \Cref{col2} implies that at least one of the most preferred alternatives of every voter receives positive probability, a property that is known as \emph{positive share} in the context of dichotomous preferences \shortcite{BMS05a,BBPS21a}. When strengthening Pareto-optimality to \sd-efficiency, we derive an impossibility (see \Cref{rem:SD}) and this impossibility can be interpreted as a strengthening of the result by \shortciteA{BBEG16a} for the subclass of support-based SDSs. 
Finally, \Cref{col3} is unrelated to the aforementioned results as the Condorcet loser property is independent of the other axioms. This result can be interpreted as a new far-reaching impossibility for SDSs.

\section{Conclusion}

We have studied which SCFs satisfy strategyproofness according to Kelly's preference extension and obtained results for three broad classes of SCFs. A common theme of our results is that strategyproofness entails that potentially ``bad'' alternatives need to be chosen. In particular, we have shown that \emph{(i)} every strategyproof rank-based SCF returns a Pareto-dominated alternative in at least one profile, \emph{(ii)} every strategyproof support-based SCF that satisfies Pareto-optimality returns at least one most preferred alternative of every voter, and \emph{(iii)} every strategyproof non-imposing SCF returns the Condorcet loser in at least one profile. All of these impossibilities rely on general insights about decisive, nominating, and vetoing groups of voters for strategyproof SCFs. Taken together, our results show that there is only room for rather indecisive strategyproof SCFs such as the Pareto rule, the omninomination rule, the SCF that returns all top-ranked alternatives that are Pareto-optimal, or the SCF that returns all alternatives except Condorcet losers. Furthermore, since we require sufficiently weak axioms, our results directly extend to randomized social choice and we therefore derive three impossibilities as corollaries for this setting.

In comparison to other results on the strategyproofness of set-valued SCFs, we employ a very weak notion of strategyproofness. In particular, our notion of strategyproofness is weaker than those used by \shortciteA{DuSc00a}, \shortciteA{BDS01a}, \shortciteA{ChZh02a}, \shortciteA{Rodr07a}, and \shortciteA{Sato08a}. This is possible because we consider the more general domain of weak preferences, which explicitly allows for ties. Interestingly, all proofs except that of Claim 1 in \Cref{thm1} can be transferred to the domain of strict preferences by carefully breaking ties and replacing Kelly-strategyproofness with the significantly stronger strategyproofness notion introduced by \shortciteA{DuSc00a}. While the resulting theorems are covered by the Duggan-Schwartz impossibility, this raises intriguing questions concerning the relationship between strategyproofness results for weak and strict preferences. 

In contrast to previous impossibilities for Kelly's preference extension \shortcite{BBGH18a,BSS19a}, our proofs do not rely on the availability of artificial voters who are completely indifferent between all alternatives.
Moreover, the results are tight in the sense that they cease to hold if we remove an axiom, reduce the number of alternatives or voters, weaken the notion of strategyproofness, or require strict preferences. For example, the essential set \shortcite{DuLa99a,Lasl00b} and a handful of other support-based Condorcet extensions satisfy strategyproofness if preferences are strict and participation for unrestricted preferences \shortcite{Bran11c,BBGH18a}. Our results thus provide important insights on when and why strategyproofness can be attained.

\acks{This work was supported by the Deutsche Forschungsgemeinschaft under grant \mbox{BR 2312/12-1}.
We thank the anonymous reviewers for helpful comments. 
A preliminary version of this article appeared in the Proceedings of the 20th International Conference on Autonomous Agents and Multiagent Systems (May 2021). Results from this article were presented at the 8th International Workshop on Computational Social Choice (June 2021).
}

\appendix

\section{Alternative Proof of \Cref{thm4}}
Subsequently, we discuss an alternative proof for \Cref{thm4} which does not rely on completely indifferent voters.
	
\CLandNI*

\begin{proof} 	
	Assume for contradiction that there is a non-imposing SCF $f$ that satisfies the Condorcet loser property and strategyproofness for $n\geq 4$ voters and $m\geq 3$ alternatives. Since the Condorcet loser property strongly depends on the parity of the number of voters, we proceed with a case distinction on $n$. For both cases, it is important that \Cref{lem:CL} and \ref{imp2} show that an alternative is the unique winner if it is uniquely top-ranked by at least $l=\lceil\frac{n+1}{2}\rceil$ voters. This follows from the observation that every set $I\subseteq N$ with $|I|\geq l$ is vetoing for $f$. Thus, \Cref{lem:CL} shows that such sets are also decisive, i.e., $f$ needs to choose an alternative $x$ as a unique winner if all voters in $I$ report the same preference relation with $x$ as unique top choice. Finally, \ref{imp2} allows to reorder the alternatives $y\in A\setminus \{a\}$ without affecting the choice set, which proves this auxiliary claim.\bigskip
	
	\textbf{Case 1: $n$ is odd}
	
	First, assume that $f$ is defined for an odd number of voters $n\geq 4$, and consider the following profiles, where $X=A\setminus \{a,b,c\}$.\smallskip
	
	\begin{profile}{C{0.05\profilewidth} C{0.16\profilewidth} C{0.28\profilewidth} C{0.27\profilewidth}}
		$R^1$: & $1$: $a, b, X, c$  & $[2\dots l]$: $a,c,X,b$ & $[l+1\dots n]$: $b,X,\{a,c\}$
	\end{profile}
	
	\begin{profile}{C{0.05\profilewidth} C{0.16\profilewidth} C{0.28\profilewidth} C{0.16\profilewidth} C{0.28\profilewidth}}
		$R^2$: & $1$: $a, b, X, c$  & $[2\dots l-1]$: $a,c,X,b$ & $l$: $c,a,X,b$ & $[l+1\dots n]$: $b,X,\{a,c\}$
	\end{profile}
	
	\begin{profile}{C{0.05\profilewidth} C{0.16\profilewidth} C{0.28\profilewidth} C{0.16\profilewidth} C{0.28\profilewidth}}
		$R^3$: & $1$: $a, b, X, c$  & $[2\dots l-1]$: $a,c,X,b$ & $l$: $c,a,X,b$ & $[l+1\dots n]$: $b,c,X,a$
	\end{profile}
	
	\begin{profile}{C{0.05\profilewidth} C{0.27\profilewidth} C{0.27\profilewidth} C{0.27\profilewidth}}
		$R^4$: &  $[1\dots l-1]$: $a,c,X,b$ & $l$: $c,a,X,b$ & $[l+1\dots n]$: $c,b,X,a$
	\end{profile}
	
	\begin{profile}{C{0.05\profilewidth} C{0.27\profilewidth} C{0.27\profilewidth} C{0.27\profilewidth}}
		$R^5$: &  $[1\dots l-1]$: $a,c,X,b$ & $l$: $c,a,X,b$ & $[l+1\dots n]$: $b,c,X,a$
	\end{profile}
	
	First, note that \Cref{lem:CL} and \ref{imp2} show that $f(R^1)=\{a\}$. Moreover, $c$ is the Condorcet loser in $R^1$ because every voter prefers $a$ weakly to $c$ and all voters in $\set{l+1}{n}$ and voter $1$ prefer all alternatives in $A\setminus \{a,c\}$ strictly to $c$. Alternative $c$ even remains the Condorcet loser if voter $l$ swaps $a$ and $c$. Hence, let $R^2$ denote the resulting profile and observe that $c\not\in f(R^2)$ because of the Condorcet loser property. In turn, strategyproofness implies that $f(R^2)=\{a\}$ if $c\not\in f(R^1)$; otherwise, voter $l$ can manipulate by reverting back to $R^1$ as he prefers $\{a\}$ to every other subset of $A\setminus \{c\}$.
	
	As the next step, we subsequently replace the preference relations of the voters $\seti{l+1}{n}$ with $b,c,X,a$. \ref{imp1} implies for each of these steps that a subset of $\{a,c\}$ is chosen if it has been chosen before the step. Since $f(R^2)=\{a\}$, we deduce that this process results in a profile $R^3$ with $f(R^3)\subseteq \{a,c\}$ . Moreover, $f(R^3)\neq \{c\}$ as otherwise voter 1 can manipulate by swapping $a$ and $b$: after this step, $b$ is uniquely top-ranked by more than half of the voters and therefore \Cref{lem:CL} and \ref{imp2} imply that it is the unique winner. Since voter 1 prefers $\{b\}$ to $\{c\}$, strategyproofness requires that $f(R^3)\in \{\{a\}, \{a,c\}\}$.
	
	Next, we discuss another derivation for $f(R^3)$ which proves that $f(R^3)\not\in \{\{a\}, \{a,c\}\}$. For this, consider the profile $R^4$ and note that $f(R^4)=\{c\}$ because more than half of the voters report $c$ as their favorite choice. Moreover, $b$ is the Condorcet loser in $R^4$ as it is uniquely bottom-ranked by the voters $\seti{1}{l}$. This even holds if the voters in $\seti{l+1}{n}$ change their preference. Thus, we let these voters swap $b$ and $c$, and the Condorcet loser property always implies for the resulting profile that $b$ is not chosen. Just as for $R^3$, strategyproofness implies then that $c$ remains the unique winner after every step because otherwise, a voter can manipulate by reverting this modification. Thus, this process results in the profile $R^5$ with $f(R^5)=\{c\}$.
	
	Finally, we derive the profile $R^3$ from $R^5$ by replacing the preference relation of voter 1 with $a,b,X,c$. Strategyproofness from $R^5$ to $R^3$ implies that $f(R^3)\neq \{a\}$ and $f(R^3)\neq \{a,c\}$ as otherwise, voter 1 can manipulate by deviating from $R^5$ to $R^3$. This is in conflict with our previous observation and hence, there is no strategyproof SCF for odd $n\geq 5$ that satisfies non-imposition and the Condorcet loser property. \medskip
	
	\textbf{Case 2: $n$ is even}
	
	As second case, we assume that $f$ is defined for an even number of voters $n\geq 4$. First, note that the induction basis in the proof of \Cref{thm4} shows that no strategyproof and non-imposing SCF for $m\geq 3$ alternatives and $n=4$ voters satisfies the Condorcet loser property, even if we forbid completely indifferent voters. The reason for this is that no such voters are required for the proof. Subsequently, we demonstrate how we can reduce the case with $n>4$ voters to the case with $n=4$ voters. Hence, assume that there is a strategyproof SCF $f$ for $n>4$ voters, $n$ even, and $m\geq 3$ alternatives that satisfies the Condorcet loser property and non-imposition. We use this SCF $f$ to define another SCF $g$ for $n=4$ voters as follows (where $X=A\setminus \{a,b,c\}$): given a profile $R$ on $4$ voters, $g$ adds $(n-4)/2$ voters whose preference relation is $c,X,b,a$ and $(n-4)/2$ whose preference relation is $a,b,X,c$. Then, $g$ returns the choice set of $f$ on the resulting profile $R'$, i.e., $g(R)=f(R')$. Subsequently, we prove that $g$ satisfies all criteria required for deriving the impossibility in the $4$ alternative case. As a consequence, $g$ cannot exist, which implies that $f$ also violates one of the required axioms. 
	
	First, note that $g$ inherits the strategyproofness of $f$ because any manipulation of $g$ is by definition also a manipulation of $f$. Moreover, $g$ cannot return the Condorcet loser because the Condorcet loser in a profile $R$ on $4$ voters is also the Condorcet loser in the profile $R'$ that is obtained after $g$ adds the $n-4$ extra voters. The reason for this is that the preference of the first half of these $n-4$ voters is inverse to the other half. In more detail, adding these $n-4$ voters increases every support $s_{xy}(R)$ by $(n-4)/2$ if $x\in \{a,b,c\}$ or $y\in \{a,b,c\}$ and the supports $s_{xy}(R)$ with $x,y\in X$ do not change at all. Consequently, the Condorcet loser does not change and $g$ inherits the Condorcet loser property from $f$. 

	The last axiom required for the proof of \Cref{thm4} is non-imposition. However, a close inspection of the proof shows that we actually do not need full non-imposition, but it is sufficient if there are three alternatives that can be chosen uniquely. Hence, we only show that $g$ can return $a$, $b$, and $c$ as unique winner. For $a$ and $c$, this follows from \Cref{lem:CL} because $g$ adds $(n-4)/2$ voters with preference $a,b,X,c$ and $(n-4)/2$ voters with preference $c,X,b,a$ to derive the input profile $R'$ for $f$. Hence, if all of the four original voters report $a,b,X,c$, then $f(R')=\{a\}$, and if all four voters report $c,X,b,a$, then $f(R')=\{c\}$. The reason for this is that in both cases, the corresponding alternative is uniquely top-ranked by more than half of the voters in $R'$ and \Cref{lem:CL} and \ref{imp2} thus show that this alternative needs to be chosen uniquely. 
	
	A slightly more complicated argument is required for showing that $g$ returns can return~$b$ as unique winner. Thus, consider the profiles $R$ and $R'$ shown below.\smallskip
	
	\begin{profile}{C{0.05\profilewidth} C{0.27\profilewidth} C{0.27\profilewidth} C{0.27\profilewidth}}
		$R$: & $[1\dots 4]$: $b, X, c,a$  & $[5\dots 2+n/2]$: $b,a,X,c$ & $[3+n/2\dots n]$: $c,X,b,a$
	\end{profile}
	
	\begin{profile}{C{0.05\profilewidth} C{0.27\profilewidth} C{0.27\profilewidth} C{0.27\profilewidth}}
		$R'$: & $[1\dots 4]$: $b, X, c,a$  & $[5\dots 2+n/2]$: $a,b,X,c$ & $[3+n/2\dots n]$: $c,X,b,a$
	\end{profile}
	
	First, note that \Cref{lem:CL} and \ref{imp2} imply that $f(R)=\{b\}$. Moreover, alternative~$a$ is uniquely bottom-ranked by all voters in ${\set{1}{4}}\cup{\set{3+n/2}{n}}$ and it thus is the Condorcet loser. This is also true if the voters $\seti{5}{2+n/2}$ swap $a$ and $b$ one after another. Hence, the Condorcet loser property implies that $a$ is not chosen after these swaps and strategyproofness entails then that $b$ is still the unique winner since all voters in $\set{5}{2+n/2}$ prefer $\{b\}$ to every other subset of $A\setminus \{a\}$. This means that $f(R')=\{b\}$. Finally, note that $g(R'')=f(R')=\{b\}$ for the profile $R''=({\succsim_1'}, {\succsim_2'}, {\succsim_3'}, {\succsim_4'})$ because the preferences of the last $n-4$ voters are equal to those used by $g$ to extend profiles consisting of $4$ voters to profiles for $n$ voters. This proves that $g$ can also return $b$ as unique winner, and thus, the proof in the main body shows that $g$ cannot exist. On the other hand, we have shown that, if there is a strategyproof SCF for an even number of voters $n>4$ that satisfies the Condorcet loser property and non-imposition, $g$ exists. By the contraposition of this implication, the impossibility generalizes to all even numbers of voters $n>4$.
\end{proof}

\section{Examples for the Tightness of our Results}

In this appendix, we discuss the SCFs that have been used to show that our results are tight. First, we deal with rank-basedness under strict preferences. Therefore, we consider the variant of the $2$-plurality rule mentioned in \Cref{rem:rankbstrict}, which we call $2^*$-plurality. For introducing this rule, we define the plurality score $\mathit{PL}(a,R)$ of an alternative $a$ in a profile $R$ as the number of voters that top-rank alternative $a$ in the profile~$R$. Given a profile $R$, let $a_R$ denote the alternative with the second highest plurality score. Then, the $2^*$-plurality rule, abbreviated by $2^*\text{-}PL(R)$, chooses precisely the alternatives $x$ with $\mathit{PL}(x,R)\geq \mathit{PL}(a_R,R)$ and $\mathit{PL}(x,R)>0$, i.e., $2^*\text{-}PL(R)=\{x\in A\colon PL(x,R)\geq PL(a_R, R) \land PL(x, R)>0\}$.

\begin{proposition}\label{prop:2PL}
	For strict preferences, the $2^*$-plurality rule is rank-based, Pareto-optimal, and strategyproof, but no voter is a nominator if $m\geq 3$ and $n\geq 5$.
\end{proposition}
\begin{proof}
	First, note that $2^*$-plurality is by definition rank-based and it satisfies Pareto-optimality as it only returns alternatives that are top-ranked by some voters. This criterion entails Pareto-optimality as we assume strict preferences. Moreover, no voter is a nominator for $2^*$-plurality if there are at least $5$ voters and at least $3$ alternatives because the top-ranked alternative $c$ of a voter can have plurality score $1$ and two other alternatives may have plurality score $2$ or more. Hence, it only remains to show that $2^*\text{-}PL$ is strategyproof. We assume for contradiction that this is not the case, i.e., that there are preference profiles $R$ and $R'$ and a voter $i$ such that ${\succsim_j}={\succsim_j}'$ for all $j\in N\setminus \{i\}$ and $2^*\text{-}PL(R')\succ_i 2^*\text{-}PL(R)$. We proceed with a case distinction on whether voter $i$'s most preferred alternative in $R$, denoted by $a$, is chosen.
	
	First, assume that $a\in 2^*\text{-}PL(R)$. This means that voter $i$ can only manipulate if $2^*\text{-}PL(R')=\{a\}$ as otherwise, there is an alternative $x\in 2^*\text{-}PL(R')$ with $a\succ_i x$. Moreover, if $2^*\text{-}PL(R)=\{a\}$, voter $i$ can also not manipulate as his best alternative is the unique winner. Hence, another alternative $b$ is chosen by $2^*$-plurality, which implies that another voter reports $b$ as his most preferred alternative in $R$. As a consequence, $PL(b,R')>0$ and therefore, $2^*\text{-}PL(R')\neq \{a\}$ as $2^*$-plurality only returns a single winner if all voters report it as their best choice. Hence, no manipulation is possible in this case.
	
	Next, assume that $a\not\in 2^*\text{-}PL(R)$ and let $b$ denote voter $i$'s best alternative in $R'$. Note that $a\neq b$ because the plurality scores, and therefore also the choice set of $2^*$-plurality, do not change otherwise. We use another case distinction with respect to the plurality score of $b$ in $R$. First, assume that $PL(b,R)\geq PL(a_R, R)>0$, which means that $b\in 2^*\text{-}PL(R)$. In particular, the claim that $PL(a_R, R)>0$ is true as $a\not\in 2^*\text{-}PL(R)$ but $PL(a, R)>0$. This means also that $2^*$-plurality elects at least two alternatives in $R$, and we choose $c\in 2^*\text{-}PL(R)\setminus \{b\}$ as the alternative with the highest plurality score in $A\setminus\{b\}$. Now, if $PL(c,R)>PL(b,R)$, alternative $b$ has the second highest plurality score in $R$, i.e., $PL(b,R)\geq PL(x,R)$ for all $x\in A\setminus \{b,c\}$. Since $PL(b,R')=PL(b,R)+1$, $PL(a,R')=PL(a,R)-1$, and $PL(x,R')=PL(x,R)$ for all $x\in A\setminus \{a,b\}$, it follows therefore that $PL(c,R')\geq PL(b, R')$ and $PL(b,R')> PL(x,R')$ for all $x\in A\setminus \{b,c\}$. Hence, $2^*\text{-}PL(R')=\{b,c\}$. On the other hand, if $PL(b,R)\geq PL(c,R)$, $b$ is the alternative with the highest plurality score in $R$, and $c$ the one with the second highest plurality score. Since $PL(b,R')=PL(b,R)+1$, $PL(a,R')=PL(a,R)-1$, and $PL(x,R')=PL(x,R)$ for all $x\in A\setminus \{a,b\}$, it follows that $PL(b, R')>PL(c,R')\geq PL(x,R')$ for all $x\in A\setminus \{b,c\}$, which shows that $c$ is also in $R'$ the alternative with the second highest plurality score. Hence, we derive also in this case that $\{b,c\}\subseteq 2^*\text{-}PL(R')$. Hence, we have in both cases that $\{b,c\}\subseteq 2^*\text{-}PL(R)\cap 2^*\text{-}PL(R')$, which contradicts that voter $i$ benefits by deviating from $R$ to $R'$ because he is not indifferent between $b$ and $c$, i.e., there are alternatives $x\in 2^*\text{-}PL(R)$, $y\in 2^*\text{-}PL(R')$ such that $x\succ_i y$. 
	
	Finally, assume that $PL(b,R)<PL(a_R,R)$ and note that this assumption entails that there are at least two alternatives with a higher plurality score than $a$ and $b$, i.e., $PL(a_R,R)>PL(a,R)$ and $PL(a_R,R)>PL(b,R)$. Hence, $PL(b,R')=PL(b,R)+1\leq PL(a_R,R)$. This means that $PL(a_R,R)=PL(a_{R'},R')$ as $b$ has a plurality score of at most $PL(a_R,R)$ in $R'$. Since the plurality scores of all alternatives $x$ with $PL(x,R)\geq PL(a_R,R)$ have not been affected by the manipulation, it follows that every alternative chosen in $2^*\text{-}PL(R)$ is also chosen after the manipulation, i.e., $2^*\text{-}PL(R)\subseteq 2^*\text{-}PL(R')$. Since $|2^*\text{-}PL(R)|\geq 2$, deviating from $R$ to $R'$ is no manipulation because we can find alternatives $x\in 2^*\text{-}PL(R)$, $y\in 2^*\text{-}PL(R)\subseteq \text{-}PL(R')$ such that $x\succ_i y$. 
	Hence, no case allows for a manipulation, which means that $2^*$-plurality is strategyproof for strict preferences. 
\end{proof}

Next, we consider \Cref{rem:rankb} in which we claim that the bounds on $n$ and $m$ in \Cref{thm1} are tight as the Pareto rule is rank-based for small values of $n$ and $m$. We prove this statement subsequently. 

\begin{proposition}\label{prop:POrb}
	The Pareto rule is rank-based, Pareto-optimal, and strategyproof if $m\leq 3$, or if $m\leq 4$ and $n\leq 2$. 
\end{proposition}
\begin{proof}
	The Pareto rule is known to satisfy Pareto-optimality and strategyproofness, regardless of the number of alternatives or voters \shortcite<see, e.g.,>{BSS19a}. Hence, it only remains to show that it also satisfies rank-basedness under the restrictions on $n$ and $m$, for which we use a case distinction.\bigskip
	
	\textbf{Case 1: $m\leq 2$}
	
	For $m=1$, rank-basedness is obviously no restriction, and if $m=2$, the rank vector of the single alternative determines all preference relations (except that we do not know which voter submits which preference relation). If an alternative $a$ is uniquely top-ranked by a voter, its rank vector contains a $(0,1)$ entry, and thus, the rank vector of the other alternative $b$ must contain a $(1,1)$ entry. Similarly, if $a$ has a $(0,2)$ entry, a voter is indifferent between both alternatives and thus, $b$ has also a $(0,2)$. Finally, we can apply a symmetric argument to the first case if the rank vector $a$ contains a $(1,1)$ entry, and thus, we can reconstruct a unique profile (up to renaming the voters) given a rank matrix. Hence, the Pareto rule is rank-based if $m=2$.\bigskip
	
	\textbf{Case 2: $m=3$}
	
	Next, we focus on the case that $m=3$ and consider an arbitrary rank matrix $Q$. First note that $Q$ can only have the following entries: $(0,3)$, $(0,2)$, $(1,2)$, $(0,1)$, $(1,1)$, and $(2,1)$. Many of these entries specify the preferences of the voters. For instance, the $(0,3)$ entry entails that a voter is completely indifferent between all alternatives. Consequently, we can add a completely indifferent voter for every $(0,3)$ entry in the rank vector of an alternative $a$, and remove all these entries from $Q$ afterwards. Also, the $(0,2)$ entries in the rank vector of $a$ specify a lot of information: there must be a voter who top-ranks $a$ and another alternative $x$ and bottom-ranks the last alternative $y$ uniquely. We use this observation to formulate a system of linear equations. Let therefore $n_a$, $n_b$, and $n_c$ denote the number of $(0,2)$ entries in the rank vector of the respective alternatives. Moreover, let $x_{ab}$, $x_{ac}$, and $x_{bc}$ denote the number of voters who top-rank both alternatives in the index. The following equations must hold for every profile $R$ with $r^*(R)=Q$. 
	\begin{align*}
	n_a&= x_{ab}+x_{ac}\\
	n_b&= x_{ab}+x_{bc}\\
	n_c&= x_{bc}+x_{ac}
	\end{align*}
	
	It can easily be checked that the unique solution of this system of equations is $x_{ab}=\frac{n_a+n_b-n_c}{2}$, $x_{bc}=\frac{n_b+n_c-n_a}{2}$, and $x_{ac}=\frac{n_a+n_c-n_b}{2}$. Since this solution is unique, these entries determine the preference relations of several voters: for instance, there must be $x_{ab}$ voters who are indifferent between $a$ and $b$, and prefer both alternatives to $c$. A symmetric argument applies also for all $(1,2)$ entries. Hence, we can now remove these entries from $Q$, as well as the corresponding $(2,1)$ and $(0,1)$ entries, to derive a reduced rank matrix. 
	
	After the last step, $Q$ only consists of $(0,1)$, $(1,1)$, and $(2,1)$ entries, which means that all remaining preferences are strict. Unfortunately, these entries do not necessarily entail a unique profile, but we can use all our observations so far to check for an arbitrary pair of alternatives $a$ and $b$ whether $a$ Pareto-dominates $b$. For this, we first construct the preferences involving ties as explained before and check whether one of the voters  prefers $b$ strictly to $a$. If this is the case, $a$ cannot Pareto-dominate $b$ and we are done. Otherwise, we consider the remaining entries in $Q$. First, if $Q$ is empty (i.e., there are no strict preference relations), we can check the Pareto-dominance by considering the preference profile constructed so far. Else, $a$ Pareto-dominates $b$ if and only if $a$ has no $(2,1)$ entry and $b$ has no $(0,1)$ entry. If $a$ has an $(2,1)$ entry or $b$ has an $(0,1)$ entry, then $a$ is uniquely last-ranked or $b$ is uniquely top-ranked by some voter, which prohibits that $a$ Pareto-dominates $b$. Conversely, if none of these entries exist, then $b$ has to be last-ranked whenever $a$ is second-ranked, and thus $a$ Pareto-dominates $b$. Since $a$ and $b$ were chosen arbitrary, we can check Pareto-dominance between alternatives only based on the rank matrix if $m=3$, which shows that the Pareto rule is rank-based in this case.\bigskip
	
	\textbf{Case 3: $m=4$ and $n\leq 2$}
	
	Finally, we show that the Pareto rule is also rank-based if $m=4$ and $n\leq2$. First, if $n=1$, it is trivial to compute the Pareto rule because only the top-ranked alternatives of the single voter are Pareto-optimal, and this information is contained in the rank matrix. Hence, we focus on the case that $n=2$ and show that $\textit{PO}(R)=\textit{PO}(R')$ for all profiles $R$, $R'$ with $r^*(R)=r^*(R')$. Given a rank matrix $Q$, we can therefore compute the Pareto rule on an arbitrary profile $R$ with $r^*(R)=Q$ as the outcome is independent of the choice of $R$. Hence, consider a profile $R$ and assume that alternative $b$ Pareto-dominates $a$ in $R$. The result follows by proving that $a$ is Pareto-dominated in all preference profiles $R'$ with $r^*(R)=r^*(R')$. 
	
	For this, let $(s_{xi}, t_{xi})=r({\succsim_i},x)$ denote the rank tuple of alternative $x$ in the preference relation of voter $i$ and note that $s_{xi}\leq s_{yi}$ if and only if $x\succsim_i y$ for all alternatives $x,y\in A$ and voters $i\in N$. We suppose subsequently that $s_{b1}\leq s_{b2}$; this is without loss of generality as we can just reorder the voters in our analysis. Next, note that the assumption that $b$ Pareto-dominates $a$ implies that $b \succsim_i a$ for all $i\in \{1,2\}$ and that this preference is strict for at least one voter. This means equivalently that $s_{bi}\leq s_{ai}$ for all $i\in I$ and that this inequality is strict for at least one voter. Now, if $s_{b1}\leq s_{b2}\leq \min_{i\in \{1,2\}} s_{ai}$, $b$ Pareto-dominates $a$ in all profiles $R'$ with $r^*(R)=r^*(R')$ because $s_{bi}\leq s_{ai}$ for all $i\in N$ in all such profiles $R'$ and one of these inequalities must be strict. 
	
	Hence, assume that $\min_{i\in \{1,2\}} s_{ai} < s_{b2}$, which means that $s_{b1}\leq s_{a1}<s_{b2}\leq s_{a2}$ because $b$ Pareto-dominates $a$ in $R$. Next, consider a profile $R'$ with $r^*(R')=r^*(R)$ such that $a\succ_i' b$ for some voter $i\in \{1,2\}$. If no such profile $R'$ exists, it is obvious that $b$ Pareto-dominates $a$ in every profile $R'$ with $r^*(R)=r^*(R')$. First, note that $a$ cannot be the uniquely most preferred alternative of voter $i$ in $R'$ because otherwise, $r^*(R)=r^*(R')$ cannot be true. Hence, there is an alternative $c\in A\setminus \{a,b\}$ such that $c\succsim_i'a$. Analogously, voter $i$ cannot uniquely bottom-rank $b$, which means that his preference relation in $R'$ is $c\succsim_i'a\succ_i'b\succsim_i'd$. Furthermore, we have that $s_{b2}\leq s_{a2}$, which means that $a$ is among the least preferred alternatives of the second voter $j$ in $R'$ because there are two alternatives that strictly dominate $b$ in $\succsim_i'$ and $r(\succsim_i',b)=r(\succsim_2,b)$. This means that $c$ Pareto-dominates $a$ in $R'$ because either $s_{b1}<s_{a1}$ or $s_{b2}<s_{a2}$, which means that either voter $i$ strictly prefers $c$ to $a$, or voter $j$ uniquely bottom-ranks $a$. Hence, $a$ is Pareto-dominated in $R'$, which proves our claim. 
\end{proof}

As last result, we discuss the SCF $f^*$ that satisfies Pareto-optimality and strategyproofness but violates support-basedness and that every voter is a nominator. As described in \Cref{remark}, this SCF chooses the maximal alternatives of a transitive dominance relation which slightly strengthens Pareto-dominance. In more detail, we say that an alternative $a$ \emph{dominates} alternative $b$ in a profile $R$ if $a$ Pareto-dominates $b$ or $n-1$ voters prefer $a$ the most while $s_{ab}(R)\geq 2$ and $s_{ba}(R)= 1$. It should be stressed that it is not required that~$a$ is uniquely top-ranked by $n-1$ voters, but only that it is among their best alternatives. Subsequently, we show that $f^*$ satisfies all axioms that we claim.

\begin{proposition}\label{prop:f*}
	The SCF $f^*$ satisfies Pareto-optimality and strategyproofness but violates support-basedness and that every voter is a nominator if $n\geq 3$.
\end{proposition}
\begin{proof}
	Before discussing the axioms, we first show that $f^*$ is a well-defined SCFs by proving that it chooses the maximal elements of a transitive dominance relation. Hence, consider an arbitrary profile $R$ and three alternatives $a,b,c$ and assume that $a$ dominates $b$ and $b$ dominates $c$. As there are two possibilities on how an alternative dominates another one (i.e., $a$ either Pareto-dominates $b$, or $s_{ab}(R)\geq 2$, $s_{ba}(R)=1$, and $n-1$ voter top-rank $a$), we proceed with a case distinction with respect to the dominance relations between $a$ and $b$ and between $b$ and $c$. First, consider the case that $a$ Pareto-dominates $b$ and $b$ Pareto-dominates~$c$. Then, $a$ Pareto-dominates $c$ as the Pareto-dominance relation is transitive.
	
	Next, consider the case that $a$ Pareto-dominates $b$ and $b$ dominates $c$ because $s_{bc}(R)\geq 2$, $s_{cb}(R)=1$, and $n-1$ voter top-rank $b$. Since every voter prefers $a$ (weakly) to $b$, it follows that $s_{ac}(R)\geq s_{bc}(R)\geq 2$, $s_{ca}(R)\leq s_{cb}(R)=1$ and that $n-1$ voters top-rank $a$. Hence, $a$ either Pareto-dominates $c$ if $s_{ca}(R)=0$ or satisfies the second dominance criterion if $s_{ca}(R)=1$. This means that the dominance relation is also in this case transitive. 
	
	As third case, assume that $b$ Pareto-dominates $c$, and that $s_{ab}(R)\geq 2$, $s_{ba}(R)=1$, and $n-1$ voters top-rank $a$. Since $b$ Pareto-dominates $c$, it follows that $s_{ac}(R)\geq s_{ab}(R)\geq 2$ and $s_{ca}(R)\leq s_{ba}(R)=1$. Hence, transitivity is also in this case satisfied. 
	
	Finally, assume that neither $a$ Pareto-dominates $b$ nor $b$ Pareto-dominates $c$, but $a$ dominates $b$ and $b$ dominates $c$. Consequently, we derive that both $a$ and $b$ are top-ranked by $n-1$ voters. However, this means that at most a single voter prefers $a$ strictly to $b$ and thus, $s_{ab}(R)\leq 1$. This contradicts that $a$ dominates $b$ and thus, this case cannot occur. Hence, the resulting dominance relation is transitive and $f^*$ is a well-defined SCF. 
	
	Next, note that $f^*$ satisfies Pareto-optimality as it is defined by a dominance relation that refines Pareto-dominance. Moreover, no voter is a nominator for $f^*$ because $f^*(R)=\{a\}$ for all profiles $R$ in which $n-1$ voters report $a$ as their uniquely best alternative. The SCF $f^*$ is also not support-based. To this end, consider the profiles $R^1$ and $R^2$, where $X=A\setminus \{a,b,c\}$, and note that $f^*(R^1)=\{a\}\neq \{a,b,c\}=f^*(R^2)$ even though $s^*(R^1)=s^*(R^2)$. \smallskip
	
	\begin{profile}{C{0.1\profilewidth} C{0.23\profilewidth} C{0.23\profilewidth} C{0.27\profilewidth}} 
		$R^1$: & $1$: $c,b,a,X$ &  $2$: $a,b, c, X$ & $[3\dots n]$: $a,b,c,X$
	\end{profile}

	\begin{profile}{C{0.1\profilewidth} C{0.23\profilewidth} C{0.23\profilewidth} C{0.27\profilewidth}}
		$R^2$: & $1$: $c,a,b,X$ &  $2$: $b,a, c, X$ & $[3\dots n]$: $a,b,c,X$ 
	\end{profile}

	Finally, it remains to show that $f^*$ is strategyproof. Assume for contradiction that this is not the case, i.e., there are preference profiles $R$ and $R'$ and a voter $i$ such that ${\succsim_j}={\succsim_j'}$ for all $j\in N\setminus \{i\}$ and $f^*(R')\succ_i f^*(R)$. Moreover, recall that $T_i(R)$ denotes voter $i$'s favorite alternatives in $R$. We proceed with a case distinction with respect to whether $T_i(R)\cap f^*(R)$ is empty or not. First, assume that $T_i(R)\cap f^*(R)$ is non-empty. This means that voter $i$ can only manipulate by deviating to $R'$ if $f^*(R')\subseteq T_i(R)$ and $f^*(R)\not\subseteq T_i(R)$. Since the dominance relation defining $f^*$ is transitive, it follows that there are alternatives $x\in T_i(R)$, $y\in f^*(R)\setminus T_i(R)$ such that $x$ dominates $y$ in $R'$ but not in~$R$. However, this is not possible. If $x$ does not Pareto-dominate $y$ in $R$, there is a voter $j\neq i$ with $y\succ_jx$ and thus, $x$ cannot Pareto-dominate $y$ in $R'$. Furthermore, since $x\succ_i y$, it follows that $s_{xy}(R)\geq s_{xy}(R')$ and $s_{yx}(R)\leq s_{yx}(R')$, and since $x\in T_i(R)$, voter $i$ can also not increase the number of voters who top-rank $x$. Consequently, since $x$ does not dominate $y$ in $R$, it does not dominate $y$ in $R'$. Hence, it follows from the transitivity of the dominance relation defining $f^*$ that $f^*(R')$ cannot be a subset of $T_i(R)$ if $f^*(R)\not\subseteq T_i(R)$, which means that no manipulation is possible in this case. 
	
	Next, assume that $T_i(R)\cap f^*(R)=\emptyset$, i.e., none of voter $i$'s best alternatives are chosen. Because at least one of voter $i$'s best alternatives is Pareto-optimal, it follows that there is a non-empty set of alternatives $B$ such that all voters $j\in N\setminus \{i\}$ top-rank all alternatives in $B$. Moreover, let $a$ denote one of voter $i$'s most preferred alternatives in $f^*(R)$ and let $b$ denote one of voter $i$'s most preferred alternatives in $B$. Observe that all alternatives $x$ with $b \succ_i x$ are Pareto-dominated by $b$ because all voters but $i$ top-rank $b$ and thus, these alternatives are not in $f^*(R)$. Moreover, it holds that $b\in f^*(R)$. Indeed, it could only be Pareto-dominated by alternatives in $B$, but it is voter $i$'s best alternative among these. Moreover, $s_{yb}(R)\leq 1$ for all $y\in A$ because $n-1$ voters top-rank $b$ and hence, it is not dominated. 
	
	As next step, we show that for all alternatives $y\in A$ with $y\succ_i a$ that $y\not\in f^*(R')$. We prove this claim by showing that there is for every such alternative $y$ an alternative $z\in B$ such that $s_{zy}(R)\geq 2$ and $s_{yz}(R)\leq 1$. This implies that $s_{zy}(R')\geq 2$ and $s_{yz}(R')\leq 1$ for all these alternatives because $y\succ_i a\succsim_i b\succsim_i z$, which means that $y\not\in f^*(R')$. Hence, assume for contradiction that there is an alternative $c\in A\setminus f^*(R)$ such that $c\succ_i a$ and $s_{xc}(R)\leq 1$ for all $x\in B$ ($s_{cx}(R)\leq 1$ must be true for all $x\in B$ since $n-1$ voters top-rank these alternatives). Since $c\not\in f^*(R)$ and $s_{xc}(R)\leq 1$ for all $x\in B$, it is Pareto-dominated by an alternative $d$; otherwise $c$ must be chosen. As a consequence of Pareto-dominance, we derive that $s_{xd}(R)\leq s_{xc}(R)\leq 1$ for all $x\in B$ and that $d\succsim_i c\succ_i a$. In particular, the last point implies that $d\not\in f^*(R)$ because of the definition of $a$ and hence, we can apply the same argument as for $c$. In more detail, by repeating this argument, we will eventually find a Pareto-optimal alternative $e$ with $s_{xe}(R)\leq 1$ for all $x\in B$ and $e\succ_i a$ because the Pareto-dominance relation is transitive. The definition of $f^*$ shows then that $e\in f^*(R)$, contradicting that $a\succsim_i x$ for all $x\in f^*(R)$. This is the desired contradiction and hence, there is for all alternatives $y\in A$ with $y\succ_i a$ an alternative $z\in B$ such that $s_{zy}(R)\geq 2$ and $s_{yz}(R)\leq 1$. This shows that no alternative with $y\succ_i a$ is in $f^*(R')$.
	
	As a consequence of the last observation, voter $i$ can only manipulate by deviating from $R$ to $R'$ if $x\sim_i a$ for all $x\in f^*(R')$ and there is an alternative $y\in f^*(R)$ with $a\succ_i y$. The latter observation implies that $a\succ_i b$ because all alternatives $x$ with $b\succ_i x$ are Pareto-dominated. By the definition of $b$, we can therefore derive a contradiction by proving that $B\cap f^*(R')\neq \emptyset$. Note for this that all alternatives in $B$ are also in $R'$ top-ranked by $n-1$ voters and thus $s_{xy}(R')\leq 1$ for all $x\in A$, $y\in B$. This means that an alternative $x\in B$ is only not chosen in $f^*(R')$ if it is Pareto-dominated. However, an alternative $x\in B$ can only be Pareto-dominated by another alternative in $B$ because for every alternative $y\in A\setminus B$, there is a voter $j\in N\setminus \{i\}$ such that $x\succ_i y$. Finally, as the Pareto-dominance relation is transitive, it follows that there is a Pareto-optimal alternative in $B$, and thus, $B\cap f^*(R')\neq \emptyset$. Altogether, this proves that $f^*$ is strategyproof. 
\end{proof}

\end{document}